\documentclass{article}

\usepackage[main, final]{neurips_2026}

\usepackage[utf8]{inputenc} 
\usepackage[T1]{fontenc}    
\usepackage{hyperref}       
\usepackage{url}            
\usepackage{booktabs}       
\usepackage{amsfonts}       
\usepackage{nicefrac}       
\usepackage{microtype}      
\usepackage{xcolor}         
\usepackage{natbib}

\usepackage{xifthen}
\usepackage{amsmath}
\usepackage{amssymb}
\usepackage{mathtools}
\usepackage{amsthm}
\usepackage{delimset}
\usepackage{enumitem}
\usepackage{tikz}
\usepackage[capitalize]{cleveref}
\usepackage{algorithm}
\usepackage{algorithmic}

\newcommand{\Payoff}{C }

\newcommand{\Utility}{U }
\newcommand{\utility}{u }
\newcommand{\State}{S }
\newcommand{\state}{s }

\newcommand{\Action}{A }

\newcommand{\action}{a }
\newcommand{\Horizon}{H }
\newcommand{\horizon}{h }
\newcommand{\Value}{V }

\newcommand{\actionsize}{n }
\newcommand{\actioncount}{N }
\newcommand{\actiondelta}{D }
\newcommand{\Real}{\mathbb{R} }
\newcommand{\Time}{T }
\newcommand{\timestep}{t }
\newcommand{\strategy}{\pi }
\newcommand{\Strategy}{\Pi }
\newcommand{\Regret}{R }

\newcommand{\loss}{l }
\newcommand{\Loss}{L }

\newcommand{\p}{\mathbb{P}}
\newcommand{\E}{\mathbb{E}}
\newcommand{\History}{Z }
\newcommand{\history}{z}

\newcommand{\transition}{P}
\newcommand{\alg}{\textsc{ALG}}
\newcommand{\ogd}{\textsc{OGD}}
\newcommand{\poly}{\mathrm{poly}}

\newcommand{\Ind}{\mathbb{I}}
\newcommand{\Var}{\mathrm{Var}}

\DeclareMathOperator*{\argmin}{arg\,min}
\DeclareMathOperator*{\argmax}{arg\,max}

\newtheoremstyle{tight}
  {3pt} 
  {3pt} 
  {\itshape} 
  {} 
  {\bfseries} 
  {.} 
  {.5em} 
  {} 
\theoremstyle{tight}
\newtheorem{theorem}{Theorem}
\newtheorem{proposition}{Proposition}
\newtheorem{lemma}{Lemma}

\newtheorem{definition}{Definition}
\newtheorem{assumption}{Assumption}

\newtheorem*{lemma*}{Lemma}

\title{Playing Markov Games Without Observing Payoffs}
\allowdisplaybreaks

\author{%
  Daniel Ablin\\
  Tel Aviv University\\
  Tel Aviv, Israel \\
  \texttt{danielablin@mail.tau.ac.il} \\
  \And
  Alon Cohen \\
  Tel Aviv University and Google Research\\
  Tel Aviv, Israel \\
  \texttt{alonco@tauex.tau.ac.il} \\
}

\begin{document}

\maketitle

\begin{abstract}
Optimization under uncertainty is a fundamental problem in learning and decision-making, particularly in multi-agent systems. 
Previously, \citet*{feldman2010playing} demonstrated the ability to efficiently compete in repeated symmetric two-player matrix games without observing payoffs, as long as the opponent’s actions are observed. 
Extending this capability to the Markovian setting remains an open problem.
In this paper, we introduce and formalize a new class of zero-sum symmetric Markov games, which extends the notion of symmetry from matrix games to the Markovian setting. 
We prove that a learner observing only the opponent's action sequence, without access to payoff information, can successfully compete against an adversary possessing complete knowledge of the game.
We formalize three distinct notions of symmetry in this domain and reveal a surprising structural hierarchy: the most holistic definitions of symmetry impose restrictive constraints that actually simplify the learning landscape, whereas the ``simplest'' definition represents the most general and challenging setting.
We provide polynomial-time algorithms for all three settings that achieve sublinear regret. Crucially, we demonstrate that despite the complex Markovian dynamics, a simple strategy of locally mimicking the opponent's actions suffices to guarantee robustness. 
This finding significantly broadens the class of games where robust learning is possible under severe informational disadvantage, proving that knowledge of the transition laws is not required to force a draw.
\end{abstract}
\section{Introduction}
In unfamiliar or competitive settings, imitation often emerges as a natural coping tactic. 
In rapidly evolving online markets, where advertising tactics can shift daily, newer companies frequently align their campaigns with those of more seasoned competitors. 
Assessing the success of a marketing effort is challenging, so inexperienced firms may gain by replicating the strategies of those with greater market knowledge. 
Lacking the time or budget for extensive research, these newcomers must act decisively to stay competitive against entrenched players. 
However, imitation carries significant risk: executed well, it can lead to outsize gains with minimal investment; executed poorly, it can result in outcomes worse than random chance. 
This problem is even greater if one considers the scenario where a company uses a long-range advertising campaign where each advertising strongly dependent on previous tactics.
To study this behavior, we frame it as a strategic interaction where the objectives of the competitor may conflict with your own. Specifically, we consider a zero-sum symmetric game.

Previously, \citet{feldman2010playing} studied repeated play of finite zero-sum symmetric matrix games.
They considered a scenario where the learner has no information about the game and does not observe any payoff. 
The learner is only aware of the available actions and observes the actions played by the adversary. 
On the other hand, the adversary may have complete prior knowledge of the payoff matrix.
They demonstrated a simple \emph{copycat strategy} that guarantees, against any opponent's strategy:
$\E[|\sum_{\timestep=1}^\Time \utility(\action^1_\timestep,\action^2_\timestep)|] = O(\actionsize\sqrt{\Time})$,
where $\Time$ is the number of rounds, $u$ is the payoff function, $\action^1_\timestep, \action^2_\timestep$ are the actions of the learner and adversary on round $\timestep$, respectively, and $\actionsize$ is the number of actions.

In our work we ask the following natural question: can the \emph{copycat strategy} of \citet{feldman2010playing} be extended to Markov games, where the current state depends on the history of actions? 
In such a setting, how can a learner effectively mimic an adversary and ensure sublinear return?
Crucially, we operate under a severe informational disadvantage. We assume the learner must navigate the complex transition dynamics without possessing prior knowledge of the game's parameters, neither the payoff function nor the underlying transition laws are known. The learner must rely solely on observing the actions played by the adversary.


To address this, our first challenge is to formalize the concept of a zero-sum symmetric Markov game. 
This requires capturing a core intuition:
\begin{center}
    \textit{\textbf{Fundamentally, symmetry implies that swapping roles and strategies perfectly swaps the players' payoffs, which, in a zero-sum game, simply negates the expected return. }}
\end{center}
Following this intuition, we explore three distinct notions of symmetry in Markov games.

When transitioning from static matrix games to Markov games, the most natural formulation of symmetry treats the game's overall expected value as the swapped payoff. This leads to our first definition, \emph{symmetry with respect to Markov policies} (MSG), where we require that the game's overall value remains symmetric over the entire horizon. Because this holistic approach does not enforce constraints on any specific local state, it is seemingly the broadest and least restrictive class of games.

Next, once value symmetry is established for Markov policies, an immediate follow-up question arises: what happens if the symmetry is for non-Markovian strategies? This brings us to our second definition, \emph{symmetry with respect to history-dependent policies} (HSG).

Finally, we consider a much more restrictive environment: enforcing symmetry locally at every single step. This leads to our third definition, \emph{per-state symmetric games} (SSG), where every individual state corresponds to a symmetric zero-sum matrix game. 
We note that to the best of our knowledge, even the restrictive SSG is the most general formulation of perfect symmetry in Markov games, as previous literature typically imposes additional restrictions on the transition dynamics.

At first glance, this progression-from holistic value symmetry (MSG) to strict local symmetry (SSG)-seems to move from broad, complex game classes to a highly restrictive case. 
Nevertheless, as we discover and structurally prove in this work, this intuitive hierarchy is completely backwards: MSGs and HSGs are special cases of SSGs!
Our result follows from carefully unfolding the global symmetry in MSGs and HSGs into a local one.

In light of this surprising result, we aim to design an efficient learning algorithm that works across all three settings.
Unlike symmetric matrix games, where a learner can simply mimic the opponent's actions, Markov games introduce complex transition dynamics that the learner must navigate.
We nevertheless provide polynomial-time algorithms for these settings while overcoming the aforementioned challenges.
To the best of our knowledge, we are the first to demonstrate tractable learning in Markov games without observing payoffs, restrictive opponents, or requiring auxiliary information.

\subsection{Related work}
\label{sec:related_work}

\textbf{Blackwell Approachability.}
Blackwell Approachability is an extension of Von Neumann's minimax theorem \citep{von1928theory} to vector-valued payoffs.
The objective is to ensure, in repeated play, that the average payoff vector converges to a desirable convex target set $K$. 
Unlike in the scalar case, fixed (oblivious) strategies may fail to guarantee convergence. 
Blackwell’s celebrated \emph{Approachability Theorem} \citep{blackwell1956analog} shows that in the repeated game setting, an adaptive strategy that reacts to the opponent’s play can still guarantee convergence to $K$.

Prior work by \cite{abernethy2011blackwell, shimkin2016online} shows that Blackwell Approachability and Online Convex Optimization are algorithmically equivalent in the context of repeated games. 
Specifically, any algorithm that achieves approachability can be transformed into an online algorithm with low regret, and vice versa.

\textbf{Symmetric Markov Games.}
Several recent papers have studied how to formalize symmetry in  Markov decision processes (MDPs) and Markov games. 
\citet{zinkevich2001symmetry} introduce a notion of symmetry in MDPs and show how it can be exploited for more effective learning in single-agent as well as multi-agent systems.
In their formulation, two states are symmetric if they have symmetric actions, and two actions are symmetric if they lead to symmetric outcomes.
They further extend this idea to symmetric MDPs, requiring both the payoff and transition functions to be symmetric.
\citet{zinkevich2006generalized} broadens this perspective to settings involving multiple agents.

Another line of work formulates symmetry at the level of payoffs in Markov games. \citet{flesch2013evolutionary} analyze a two-player Markov game in which the payoff functions are symmetric: if the players’ actions at a given state are exchanged, then their payoffs are swapped as well. 
A similar approach is taken by \citet{yongacoglu2023satisficing}, who consider a multi-player setting and require invariance under arbitrary permutations of the action sets. 
Both frameworks are flexible enough to capture zero-sum games. 
A key limitation, however, is that they assume transition dynamics are insensitive to the ordering of players’ actions, i.e., swapping the players does not alter the state dynamics, which restricts the generality of the resulting games.

\textbf{Learning in Markov games.}
Markov games pose significant theoretical and algorithmic challenges, particularly in adversarial and partially observable settings. 
Prior work identified strong hardness results for core learning objectives in this domain. 
For instance, \citet{bai2020near} prove that computing a best response policy in zero-sum Markov games when facing an adversarial opponent is at least as difficult as the well-known parity with noise problem, which is widely believed to be computationally intractable. 

A closely related special case that has been extensively studied is that of MDPs with known transitions but adversarial or unknown rewards. Such an MDP can be interpreted as a Markov game in which the adversary acts as a second player whose “actions” either modify the state payoffs or, equivalently, redirect transitions to payoff-altered versions of the next state. This line of research (e.g., \citealp{even2009online, zimin2013online, jin2020simultaneously}) primarily focuses on scenarios where the payoff function may vary over time.
In a related line of work, \citet{tian2021online,liu2022learning} establish tight lower bounds on regret in the online learning setting, showing that no algorithm can achieve better than $\Omega(\min\{\sqrt{2^\Horizon \Time}, \Time\})$ regret in episodic zero-sum Markov games with unknown transitions, even when the adversary is limited to using Markov strategies.

Despite these limitations, recent advances have uncovered tractable settings under more cooperative assumptions. For example, \citet{erez2023regret,cai2024near,cui2023breaking,yang2022t,jin2021v} design decentralized algorithms for Markov games that achieve sublinear regret, assuming all players employ the same proposed learning rule.
Their approach requires no communication between agents and guarantees convergence to a (coarse) correlated equilibrium.

\section{Preliminaries}

\subsection{Markov games}
A two-player zero-sum finite horizon Markov game is defined by the quintuple $(\State, \Action, \Horizon, \transition, \utility)$, where $\Horizon \geq 1$ is the time horizon of the game.
$\State$ is a finite set of states partitioned as $\State=\bigsqcup_{\horizon=1}^{\Horizon}\State_\horizon$, where $\State_1=\{\state_1\}$.
Without loss of generality, we assume that all states in $\State$ are reachable from the initial state $\state_1$ under some pair of policies, as unreachable states do not affect the game's value or dynamics.
$\Action = \{1,\dots,\actionsize\}$ is a finite set of actions available to both players; $\transition(\cdot\mid\state_\horizon,\action^1,\action^2) \in \Delta(\State_{\horizon+1})$ is the probability distribution over the next state when the players 1 and 2 take actions $\action^1$ and  $\action^2$ respectively at state $\state_\horizon \in \State_\horizon$. 
Finally, $\utility: \State \times \Action \times \Action \rightarrow [-1,1]$ is the payoff function for player 1. 
As the game is zero-sum, $-\utility$ is the payoff function for player 2.

At every time step $\horizon=1,\dots,\Horizon$, the players observe the current state $\state_\horizon$, then choose actions $\action^1_\horizon,\action^2_\horizon$ respectively. The players choose their actions independently at random according to their policies.
Then player 1 gains $\utility(\state_\horizon,\action^1_\horizon,\action^2_\horizon)$ and player 2 gains $-\utility(\state_\horizon,\action^1_\horizon,\action^2_\horizon)$, and the game transitions to the next state $\state_{\horizon+1}$ drawn from $\transition(\cdot\mid\state_\horizon,\action^1_\horizon,\action^2_\horizon)$. 
The game ends after time step $\Horizon$.

A history $\history_\horizon = (\state_1, \action_1^1, \action_1^2, \dots, \state_{\horizon-1},\action_{\horizon-1}^1, \action_{\horizon-1}^2, \state_\horizon)$ is the sequence of states and actions generated by the players up to time $\horizon$. We denote by $\History_\horizon$ the set of histories up to time $\horizon$.
A (history-dependent) policy for player $i$ is a function $\strategy^i : \History_\horizon \mapsto \Delta(A)$ that maps each history $\history_\horizon$ to a probability distribution from which player $i$ draws their action when observing $\history_\horizon$.
A policy is Markov if it is a function of only the current state, namely $\strategy^i : \State \mapsto \Delta(A)$.

We also introduce the concept of the \emph{value function}, which captures the expected cumulative payoff from a given state onward and will be fundamental in defining optimality and proving performance guarantees.
For policies $\strategy^1, \strategy^2$ of the two players, time step $\horizon$ and state $\state \in \State_\horizon$, the \emph{value function} is defined as:
$
\Value^{\strategy^1,\strategy^2}(\state) 
\coloneqq 
\E^{\strategy^1,\strategy^2}\brk[s]*{\sum_{\horizon'=\horizon}^\Horizon\utility(\state_{\horizon'}, \action_{\horizon'}^1,\action_{\horizon'}^2) ~\bigg\vert~ \state_\horizon = s},
$
where the expectation is taken over the randomness in the actions taken by the policies and the transition dynamics.
For Markov policies, the value function satisfies the Bellman equations (see, e.g., \citealp{puterman1990markov}):
\begin{equation*}
    \Value^{\strategy^1,\strategy^2}(\state) 
    =
    \E^{\strategy^1,\strategy^2} \Bigg[ \utility(\state,\action^1_\horizon,\action^2_\horizon) \\
    + \sum_{\state' \in \State_{\horizon+1}} \transition(\state' \mid \state,\action^1_\horizon,\action^2_\horizon) \Value^{\strategy^1,\strategy^2}(\state') ~\bigg\vert~ \state_\horizon=\state\Bigg]. 
\end{equation*}

We define the value of the game as $\Value^\star = \max_{\strategy^1} \min_{\strategy^2} \Value^{\strategy^1,\strategy^2} (\state_1)$.
The value is guaranteed by a Markov safety level policy $\strategy^\star$ for player 1, that is $\Value^{\strategy^\star,\strategy^2} (\state_1) \ge \Value^\star$ for any strategy $\strategy^2$ for player 2.
Moreover, there is a polynomial-time dynamic programming algorithm that returns such $\strategy^\star$ \citep{wal1976markov}.

\subsection{Online linear optimization} \label{sec:olo}
Online linear optimization is a sequential game between a learner and the environment played for $\Time$ rounds.
Before the game begins, the environment privately generates a sequence of linear loss functions $l_1,\dots,l_\Time : \Utility \mapsto \Real$ where $\Utility \subseteq \Real^m$ is a compact convex set. 
At every round $\timestep$, the learner predicts $\utility_\timestep \in \Utility$. 
Only then, the environment reveals $l_\timestep$ and the learner suffers $l_\timestep(u_\timestep)$.
The regret of the learner is defined as: $
    \Regret(\Time) 
    \coloneqq 
    \sum_{\timestep=1}^\Time l_\timestep(\utility_\timestep)
    -
    \min_{\utility \in \Utility} \sum_{\timestep=1}^\Time l_\timestep(\utility).
$
A learning algorithm is said to be consistent if its regret grows sub-linearly with time, i.e., $\Regret(\Time) = o(\Time)$. 
This ensures that, on average, the algorithm performs asymptotically as well as the best fixed strategy in hindsight.

\paragraph{Online Gradient Descent (OGD).}
OGD refers to the online optimization algorithm that starts from some $\utility_1 \in \Utility$, then on every round $\timestep$ updates: (1) $\utility_{\timestep+1}' = \utility_\timestep -\eta \,l_t$ where $\eta >0$ is the learning rate; (2) sets $\utility_{\timestep+1}$ to be the Euclidean projection of $\utility_{\timestep+1}'$ onto $U$.
For a specific choice of $\eta$, OGD achieves regret bounded by $O(GD\sqrt{\Time})$ when $\|l_\timestep\| \le G$ for all $\timestep=1,\dots,\Time$ and $\max_{\utility,\utility' \in U} \|\utility-\utility'\| \le D$ (see, e.g., \citealp{hazan2022introduction}).
Note that OGD operates in time complexity of $O(\Time (m + Q))$, where $Q$ is the time to compute the projection on $\Utility$.
Thus if $Q$ is polynomial in the problem parameters, the OGD runs in polynomial time as well.

\section{Problem definition and main results}
\label{sec:results}
We consider the problem of learning a two-player zero-sum Markov symmetric game by repeated play. We operate under a severe informational disadvantage: the learner never observes the game's payoff, only the adversary's actions, and assumes no prior knowledge of the transition probability function $\transition$.
In contrast, we assume the adversary can be omniscient, possessing complete prior knowledge of both the transition dynamics $\transition$ and the payoff function $\utility$.

The game is played for $\Time$ episodes.
At the beginning of each episode $\timestep$, both players choose policies $\strategy_\timestep^1, \strategy_\timestep^2$, that are  then used to play the game in the current episode. 
While the opponent can choose any, possibly history-dependent, policy, we restrict the learner to select only Markov policies to facilitate efficient learning. 
Thereafter, both players observe the trajectory sampled from the game's dynamics and the actions of their opponent. 
Crucially, only the adversary can calculate the incurred payoff.
The goal of the learner is to guarantee that their own payoff approaches the value of the game, which is 0. 

Formally, the learner minimizes $\E\brk[s]{\abs{\sum_{\timestep=1}^\Time \Value^{\strategy_\timestep^1, \strategy_\timestep^2}(\state_1)}}$.
We aim for a bound that is sublinear in $\Time$, implying the average payoff converges to zero.
The concept of a zero-sum symmetric Markov game, at its core, must capture the intuition that swapping roles and strategies perfectly swaps the players' payoffs, which, in a zero-sum game, simply negates the expected return.
Following this intuition, we consider three types of symmetry in Markov games.
We begin by formalizing the most natural and holistic extension of symmetry to the Markovian setting. By treating the game in its entirety as symmetric, we require that the expected cumulative return perfectly reflects role-swapping over the space of Markov policies.
\begin{assumption}[Symmetry w.r.t. Markov policies; MSG]
\label{assum:MSG}
    A zero-sum Markov game is MSG if for every pair of Markov policies $\strategy^1,\strategy^2$ it holds that
    $
        \Value^{\strategy^1, \strategy^2}(\state_1) = -\Value^{\strategy^2, \strategy^1}(\state_1).
    $
\end{assumption}

Next, we consider the natural generalization of the MSG definition, where symmetry must hold even for history-dependent policies.

\begin{assumption}[Symmetry w.r.t. history-dependent policies; HSG]
\label{assum:HSG}
    A zero-sum Markov game is HSG if for every two policies $\strategy^1,\strategy^2$, possibly history-dependent, it holds that
    $
        \Value^{\strategy^1, \strategy^2}(\state_1) = -\Value^{\strategy^2, \strategy^1}(\state_1).
    $
\end{assumption}
Finally, while MSG and HSG enforce symmetry on the overall value of the game, one might ask: what if we require each individual state to be symmetric?
\begin{assumption}[Per-state symmetric game; SSG]
\label{assum:SSG}
    A zero-sum Markov game is SSG if $\utility(\state,\action^1,\action^2) = -\utility(\state,\action^2,\action^1)$ for all $\state \in \State$ and $\action^1,\action^2\in \Action$.
    Namely, the payoff matrix in each state corresponds to a zero-sum symmetric matrix game. 
\end{assumption}

\vspace{-1em}
\subsection{Our results}
Our analysis reveals a counter-intuitive hierarchy across the three symmetry classes. One might naturally assume that ``holistic'' notions of symmetry (MSG and HSG), which enforce symmetry over policies rather than local states, would yield a richer, more general class of games.
However, we establish the exact opposite: holistic symmetry requirements impose such severe constraints on the game structure that they effectively simplify the learning landscape.
We detail this hierarchy below.

First, for \textbf{MSG} environments (\cref{assum:MSG}), defining symmetry holistically over full policies introduces complex temporal dependencies that mathematically couple the entire state space. This drastically increases the complexity of the game, rendering direct optimization or equilibrium analysis computationally prohibitive, if not entirely intractable.


To overcome this, we develop a structural reduction that disentangles the game's complex transition dynamics. 
We show that by identifying the correct transformation, we can recursively ``peel back'' the horizon, revealing that every layer of an MSG effectively decomposes into independent skew-symmetric matrix games. 
We demonstrate that this transformation preserves the game's value for any pair of strategies. 
%
Formally, we prove the following:
\begin{theorem}
\label{prop:msg}
Any MSG can be efficiently transformed into an SSG by recursively applying the Expected Payoff Redistribution (EPR) reduction, without changing the expected cumulative payoff for any pair of policies.\footnote{Note that the EPR transformation accumulates payoffs over time steps, which may inflate them by up to a factor of $\Horizon$.}
\end{theorem}

This result (detailed in \cref{sec:MSG}) implies that, in practice, the learner can simply act as if they are playing against the modified SSG. By doing so, they seamlessly preserve all subsequent baseline learning guarantees while entirely bypassing the complex, temporally coupled dynamics of the original MSG.




Next, for \textbf{HSG} environments (\cref{assum:HSG}), we demonstrate that the symmetry constraint imposes a strong structural restriction on the game.
In particular, while the game may continue through arbitrary stochastic transitions as long as both players choose identical actions, once the players choose different actions, their subsequent choices cannot affect the expected remaining payoff.
Formally, we prove the following:
\begin{theorem}
\label{prop:hsg}
Any HSG can be transformed into a value-equivalent SSG that terminates upon the first disagreement between the players.
\end{theorem}

This result (\cref{sec:HSG}) reveals that HSGs are structurally equivalent to games that terminate upon the first disagreement between the players.
Importantly, the game may still involve complex stochastic transitions and multiple states before this disagreement. Consequently, the transition-free learning algorithm developed for SSGs also applies to HSGs.

Finally, for \textbf{SSG} environments (\cref{assum:SSG}), we establish the following baseline guarantee.

\begin{theorem} 
\label{prop:ssg}
For any SSG, we present a learning algorithm that guarantees $\E\brk[s]{\abs{\sum_{\timestep=1}^\Time \Value^{\strategy_\timestep^1, \strategy_\timestep^2}(\state_1)}}=O(\actionsize \sqrt{\Time|\State|\Horizon})$ in $\poly(|\State|, \actionsize, \Time)$ run-time.
\end{theorem}

This result (detailed in \cref{sec:SSG}) relies on a reduction to online linear optimization via Blackwell approachability.
Moreover, our algorithm is implemented efficiently using the observation that local state symmetry allows the learner to decouple the Markov game.

To complement these upper bounds, we provide a matching lower bound for SSG games in \cref{app:lower_bound}.

\begin{proposition}
\label{prop:lower_bound}
There exists an adversary and a payoff function for which $\E\brk[s]{\abs{\sum_{\timestep=1}^\Time \Value^{\strategy_\timestep^1, \strategy_\timestep^2}(\state_1)}}=\Omega(\actionsize \sqrt{\Time|\State|\Horizon})$ against any, possibly randomized, learner.
This result confirms that our algorithms operate in the optimal regime with respect to the parameters of the game.
\end{proposition}

Finally, to complete our theoretical analysis, we demonstrate the robustness of our framework by extending our results to general-sum games (\cref{sec:general_sum}) and environments with $\epsilon$-approximate symmetry (\cref{app:approximate-symmetry}).

\textbf{Significance of results.} To the best of our knowledge, this work is the first to define symmetry holistically over the entire Markov game rather than restricting it to local states. We adopt this holistic approach because it faithfully captures the true underlying intuition of a zero-sum symmetric game: that swapping roles and strategies should perfectly negate the expected return. 

At first glance, MSG may appear to define a broader and more challenging class than SSG, introducing local state asymmetries and mathematically coupled global constraints that render direct optimization computationally prohibitive, let alone when extended to HSG. 
This work reveals a counter-intuitive reality: our structural reductions (\cref{prop:msg,prop:hsg}) prove that attempting to induce complexity through these holistic definitions neutralizes the impact of the transition dynamics. 
By collapsing these temporally coupled games into independent matrix games, we bypass highly complex algorithmic constructions. 
Consequently, our findings suggest that SSG is not merely the simplest definition, but structurally the most general.


This structural collapse profoundly expands the boundaries of symmetric MARL research. 
In the broad literature on finding equilibria in symmetric zero-sum games, the standard requirement for establishing a Markov Perfect Equilibrium (MPE) is strict, local, per-state symmetry for both payoffs and transitions (which constitutes a subset of our SSG). 
Nevertheless, our work reveals that SSG is arguably the most general definition of symmetry utilized in Markov game literature, and our results suggest that any prior work relying on stricter constraints might be able to drastically relax their assumptions.
By porting their results through our transformation, the entire field's definition of symmetry can be potentially broadened to encompass holistic policy symmetry (e.g., MSG, HSG) while preserving fundamental equilibrium or regret guarantees.
In summary, our structural collapse is not just a stepping stone for regret bounds, but also provides a powerful blueprint for simplifying complex sequential decision-making problems.

\vspace{-1em}
\section{Copycat strategy in MSG}
\label{sec:MSG}

In this section, we prove \cref{prop:msg}. 
Our approach to learning under the MSG assumption (\cref{assum:MSG}) focuses on exploiting the structural constraints induced by the symmetry of the value function.

The key distinction from SSG is that symmetry is now enforced on the expected cumulative payoff rather than locally at each state. The set of valid MSG payoff functions is defined as:
\begin{align*}
    \Utility 
    =  
    \Bigg\{\utility \in \Real^{|\State| \times \actionsize \times \actionsize} 
    ~:~ 
    |\utility(\cdot)| \le 1,
     ~\text{and}~
    \Value^{\strategy^1,\strategy^2}(\state_1)=-\Value^{\strategy^2,\strategy^1}(\state_1)
    , \forall \strategy^1,\strategy^2 \in \Strategy_{M}\Bigg\},
\end{align*}
where $\Strategy_{M}$ denotes the set of all Markov policies.

The central challenge in the MSG setting is that while the game is globally zero-sum symmetric, individual states need not be skew-symmetric. This breaks the local guarantee of the Minimax strategy used in the matrix case.
%
Intuitively, one might expect the MSG assumption to be a broad generalization of SSG. However, we demonstrate that this intuition is misleading: the holistic nature of MSG imposes constraints that some SSG games fail to satisfy. 
We illustrate this distinction with the counter-example in \cref{sec:SSG-not-MSG}.

We leverage this insight by analyzing the game from the latest timestep, by fixing the strategy up to the final state.
This motivates our recursive reduction technique.

\subsection{Expected Payoff Redistribution}
\label{sec:EPR}

We introduce the \textbf{Expected Payoff Redistribution (EPR)} transformation, which propagates expected future rewards into the current state to reveal the underlying symmetry of the future states.
\begin{definition}[Expected Payoff Redistribution; EPR]
\label{def:EPR}

Given a payoff function $\utility$ with horizon $\Horizon \geq 2$, the EPR transformation on round $\horizon$ yields a new payoff $\hat\utility$ defined as:
\begin{itemize}[nosep,leftmargin=*]
        \item       $\hat\utility(\state,\action^1,\action^2) 
        = \utility(\state,\action^1,\action^2) 
        + \sum_{\state'\in\State_{\horizon+1}} \p(\state'\mid\state,\action^1,\action^2)\utility(\state',1,1)$ for all $\state \in \State_{\horizon}$.
        
        \item 
        $\hat\utility(\state,\cdot,\cdot) = \utility(\state,\cdot,\cdot) - \utility(\state,1,1)$ for all $\state \in \State_{\horizon+1}$.
        \item 
        $\hat\utility(\state,\cdot, \cdot) = \utility(\state,\cdot, \cdot)$ otherwise.
\end{itemize}
\end{definition}
\begin{lemma}[EPR on MSG]
\label{lem:shorter-game}

    Consider an MSG game with $\Horizon \geq 2$. Applying the EPR transformation (\cref{def:EPR}) on round $\Horizon-1$ yields $\hat\utility$ with the following properties:
    \begin{enumerate}[nosep,leftmargin=*]
        \item At round $\Horizon$, the payoff matrix for every state is skew-symmetric. \label{it:last-round-symmetric}
        \item If we truncate the last round, the resulting game of length $\Horizon-1$ satisfies the MSG assumption (\ref{assum:MSG}). \label{abs}
    \end{enumerate}
\end{lemma}
\vspace{-1em}
\begin{proof}[Proof Sketch]
    The full proof (deferred to \cref{app:shorter-game-proof}) relies on mathematically ``calibrating'' the game using a reference action pair. Because the MSG assumption mandates that $V^{\pi^1,\pi^2}(s_1) = -V^{\pi^2,\pi^1}(s_1)$ for \emph{all} Markov policies, we can isolate the final step's payoff structure by evaluating specific, paired strategies. Specifically, we construct a baseline strategy where both players deterministically play a fixed action (e.g., action $1$) upon reaching any state $s \in S_H$, while playing uniformly at random in all preceding states. By substituting these policies into the MSG condition and measuring the value of arbitrary actions at step $H$ relative to this baseline, we effectively subtract out the influence of the complex transition history.

    This algebraic isolation reveals a rigid structure in the final layer: for all actions, $u(s_H, a^1, a^2) + u(s_H, a^2, a^1) = 2u(s_H, 1, 1)$. While this final payoff matrix is not inherently skew-symmetric, it becomes strictly skew-symmetric once we subtract the baseline value of playing $(1, 1)$. The Expected Payoff Redistribution (EPR) transformation mechanically executes this by explicitly subtracting the $(1, 1)$ baseline from step $H$ and redistributing it backward to step $H-1$. This completely removes the ``bias'' at step $H$, forcing its local payoff matrix into a skew-symmetric form, while successfully propagating the expected future value backward without violating the overall MSG assumption for the truncated $H-1$ game.
\end{proof}

\begin{lemma}
\label{lem:EPR-Preserve-Value}
    For any payoff function $\utility$ with $\Horizon \geq 2$, applying EPR (\cref{def:EPR}) at any round $\horizon < \Horizon$ preserves the value of the game. Specifically, for any strategies (possibly history-dependent) $\strategy^1, \strategy^2$:   
    $
        \Value^{\strategy^1,\strategy^2}(\state_1)=        \hat{\Value}^{\strategy^1,\strategy^2}(\state_1),
    $
    where $\hat\Value$ is the value function under $\hat\utility$.
\end{lemma}
\begin{proof}
    Since $\hat\utility$ differs from $\utility$ only at steps $\horizon$ and $\horizon+1$, it suffices to show that value conservation holds locally. For any $\history \in \History_{\horizon}$ ending at state $\state$, let $\history' = \history \oplus(\action^1,\action^2,\state')$ for each $\state'\in\State_{\horizon+1}$ :
    \begin{align*}
        \Value^{\strategy^1,\strategy^2}(\history)
        &= \E^{\strategy^1,\strategy^2}\Bigg[\utility(\state, \action^1,\action^2) + \sum_{\state'\in\State_{\horizon+1}}\p(\state'\mid\state,\action^1,\action^2)\Value^{\strategy^1,\strategy^2}(\history')\Bigg]\\
        &=\E^{\strategy^1,\strategy^2}\Bigg[\utility(\state, \action^1,\action^2) + \sum_{\state'\in\State_{\horizon+1}}\p(\state'\mid\state,\action^1,\action^2)\brk{\utility(\state',1,1)+\hat\Value^{\strategy^1,\strategy^2}(\history')}\Bigg]\\
        &=\E^{\strategy^1,\strategy^2}\Bigg[\hat\utility(\state, \action^1,\action^2)+ \sum_{\state'\in\State_{\horizon+1}}\p(\state'\mid\state,\action^1,\action^2)\hat\Value^{\strategy^1,\strategy^2}(\history')\Bigg]
        =\hat\Value^{\strategy^1,\strategy^2}(\history). 
        \qedhere
    \end{align*}
\end{proof}
\vspace{-1em}
\subsection{Reducing MSG to SSG.}
\begin{lemma}
\label{lem:MSG_is_SSG}
    For any MSG game with payoff $\utility$, there exists an SSG game with payoff $\hat\utility$ such that for all policies (possibly history-dependent):
    $
        \Value^{\strategy^1,\strategy^2}(\state_1)= \hat{\Value}^{\strategy^1,\strategy^2}(\state_1).
    $
\end{lemma}
\vspace{-1em}
\begin{proof}
    We prove this by recursive reduction. We start by applying the EPR transformation (\cref{def:EPR}). \cref{lem:shorter-game} guarantees that the transformed game, when viewed as a game of horizon $\Horizon-1$, retains the MSG property. Simultaneously, it ensures the payoff at the discarded $\Horizon$-th time step becomes strictly skew-symmetric.

    We repeat this process $\Horizon-1$ times.
    At each step, we \textbf{decouple} the final time step, transforming its local payoff matrices into skew-symmetric form, and \textbf{propagate} the residual asymmetry to the preceding layer. 
    \cref{lem:EPR-Preserve-Value} ensures that the value function remains invariant throughout. 
    In the final step, the remaining payoff matrix at $\state_1$ must itself be skew-symmetric to satisfy the MSG assumption (\ref{assum:MSG}). 
    The result is a fully skew-symmetric game (SSG) that is equivalent to the original MSG.
\end{proof}
\vspace{-1em}
Consequently, from the learner's perspective, the MSG class is structurally equivalent to SSG. The learner can apply the SSG algorithm directly to any MSG instance without modification.
Building on this foundational reduction, the next section formally analyzes the SSG setting (\cref{assum:SSG}).

\section{Copycat strategy in SSG}
\label{sec:SSG}

In this section, we prove \cref{prop:ssg} and present our strategy for learning in Per-state Symmetric Games (SSG) under \cref{assum:SSG}. In this setting, every state $\state \in \State$ encapsulates a symmetric zero-sum matrix game. While \cref{sec:copycat_on_each_step} details a specialized alternative approach that runs the matrix copycat algorithm simultaneously at every state, we focus here on a broader, general framework.


Notably, in \cref{app:Blackwell’s-approachability-formulatio} we naturally frame the copycat problem in matrix games \citep{feldman2010playing} as a Blackwell approachability problem, efficiently solvable via online optimization. 
Here, we generalize this scheme to Markov games. Our formal analysis relies on direct reductions to online optimization rather than explicitly invoking Blackwell approachability, which drastically simplifies our presentation. 
We now formally analyze the SSG setting.

In the copycat problem, the learner does not observe the realized payoffs but knows that the payoff function belongs to a specific set $\Utility$. The objective is to guarantee performance comparable to the value of the game which is always zero, regardless of the chosen payoff.
By leveraging the equivalence mentioned above and viewing the problem through the lens of online optimization, we  derive a general solution as \cref{alg:Main} below.

\begin{algorithm}[H]
   \caption{Playing Games with Unobserved Payoffs}
   \label{alg:Main}
\begin{algorithmic}[1] 
    \STATE \textbf{input}: learning rate $\mu > 0$.
    \STATE \textbf{init} $\utility_1(\state,\action^1,\action^2)=0 \;\ \forall \state\in \State\;\ \text{and} \;\ \action^1,\action^2 \in \Action$.
    \FOR{$\timestep=1$ {\bfseries to} $\Time$}
        \STATE \textbf{compute} $\strategy_\timestep^1 = \argmax_{\strategy^1} \min_{\strategy^2} \E^{\strategy^1,\strategy^2}\brk[s]*{\utility_t(\state,\action^1,\action^2)}$
        \STATE \textbf{play} $\strategy_\timestep^1$ and
        observe trajectory $(\state_{1,\timestep},\action_{1,\timestep}^1,\action_{1,\timestep}^2,\dots,\state_{\Horizon, \timestep},\action_{\Horizon,\timestep}^1,\action_{\Horizon,\timestep}^2)$.
        \STATE $\utility_{\timestep+1} \gets \utility_{\timestep}$.
        \FOR{$\horizon=1$ {\bfseries to} $\Horizon$}
        \STATE \textbf{update} $\utility_{\timestep+1}(\state_{\horizon,\timestep},\action^1,\action^2) = \begin{cases}
            \utility_{\timestep}(\state_{\horizon,\timestep},\action^1,\action^2) - \frac{\mu}{2}, & \action^1=\action_{\horizon,\timestep}^1,\action^2=\action_{\horizon, \timestep}^2\; , \action^1_{\horizon,\timestep}\neq \action^2_{\horizon, \timestep}; \\
            \utility_{\timestep}(\state_{\horizon,\timestep},\action^1,\action^2) + \frac{\mu}{2}, & \action^1=\action_{\horizon,\timestep}^2,\action^2=\action_{\horizon, \timestep}^1 \; , \action^1_{\horizon,\timestep}\neq \action^2_{\horizon, \timestep}; \\
            \utility_{\timestep}(\state_{\horizon,\timestep},\action^1,\action^2), & \text{otherwise.}
        \end{cases}$
        \ENDFOR
        \STATE \textbf{clip} $\utility_{\timestep+1}(\state,\action^1,\action^2)$ entry-wise to $[-1,1]$ for all $\state \in \State$, $\action^1,\action^2 \in \Action$. 
    \ENDFOR
\end{algorithmic}
\end{algorithm}

We conceptualize the algorithm as involving a third ``virtual player'' who observes the interaction between the learner and the adversary.
At the beginning of each round $\timestep$, this virtual player selects a payoff function $\utility_\timestep \in \Utility$ with the goal of minimizing their cumulative loss. 
Given a realized trajectory $\tau=(\state_1,\action_1^1,\action_1^2,\dots,\state_\Horizon,\action_\Horizon^1,\action_\Horizon^2)$, the loss for episode $\timestep$ is defined as the total payoff under the selected function: $\sum_{\horizon=1}^\Horizon\utility_\timestep(\state_\horizon,\action_\horizon^1,\action_\horizon^2)$

To optimize this objective, the virtual player employs an Online Gradient Descent algorithm over the set $\Utility$.
At the end of round $\timestep$, $\ogd$ receives a feedback loss tensor $\loss_\timestep \in \mathbb{R}^{|\State| \times \actionsize \times \actionsize}$.
This tensor acts as a sparse indicator of the trajectory: its entries are set to $1$ for the $\Horizon$ state-action triplets visited in $\tau$ and $0$ otherwise.
In this manner we represent the loss as a linear function of the payoff matrix: $l_\timestep \cdot \utility_\timestep = \sum_{\horizon=1}^\Horizon\utility_\timestep(\state_\horizon,\action_\horizon^1,\action_\horizon^2)$.

\begin{proof}[Proof Sketch]
The proof of the $O(\actionsize \sqrt{\Time|\State|\Horizon})$ bound is detailed in \cref{sec:online-proof}. The analysis proceeds naturally from the observation that the safety-level strategy ensures a non-negative expected payoff at each individual state, and consequently, across the entire trajectory of the game. Combining this structural property with the regret bound of $\ogd$ establishes the final guarantee.
Crucially, implementing this approach for SSG is highly computationally efficient. The per-step complexity remains strictly polynomial: because SSG enforces symmetry independently at each state, restricting the expected further value of any state to zero, both the safety-level policy computation and the projection operate over uncoupled local action spaces, making them trivial to compute.
\end{proof}

\subsection{General sum case}
\label{sec:general_sum}

We further extend our results to general-sum symmetric games, by considering the normalized ``skew-symmetrized'' payoff $\hat\utility(\state,\action^1,\action^2) = \frac{1}{2}(\utility_\star(\state, \action^1,\action^2) - \utility_\star(\state, \action^2,\action^1))$ , such that $\hat\utility(\state,\action^1,\action^2) \in [-1,1]$. 
The learner's payoff with respect to this transformed game is similarly bounded:
\begin{align*}
    \E\brk[s]*{\abs*{\sum_{\timestep=1}^\Time\sum_{\horizon=1}^\Horizon\utility_\star(\state_{\horizon,\timestep}, \action^1_{\horizon,\timestep},\action^2_{\horizon,\timestep})-\utility_\star(\state_{\horizon,\timestep},\action^2_{\horizon,\timestep},\action^1_{\horizon,\timestep})}}
    \leq\sqrt{2 \Time|\State|\Horizon\actionsize^2}.
\end{align*}

This implies a powerful result: that even in general-sum SSG, the learner can enforce a symmetric outcome on average regardless of the opponent's strategy. 
Namely, even if the adversary plays sub-optimally or maliciously, the learner effectively ``mirrors'' the opponent's success (or failure), ensuring they are never exploited.

Next, we explore learning under the even stricter condition of \cref{assum:HSG}.

\section{Copycat Strategy in HSG}
\label{sec:HSG}

In this section, we prove \cref{prop:hsg}.
We show that the HSG assumption (\cref{assum:HSG}) imposes a strong structural restriction: once the players choose different actions, their subsequent choices cannot affect the expected remaining payoff.
Consequently, after payoff redistribution, the game is equivalent to an SSG that terminates upon the first disagreement between the players.

Just as not every SSG is an MSG, the HSG class is a strictly smaller subset of MSG.
While HSGs may still involve complex stochastic transitions, the requirement for symmetry over \textit{history-dependent} policies severely restricts the strategic influence of future actions.
We illustrate this distinction with the counterexample in \cref{sec:msg-not-hsg}.

For a history $\history_h$, we denote by
$\Value^{\strategy^1,\strategy^2}(\history_h)$
the expected cumulative payoff from time step $h$ onward,
starting from $\history_h$ and following policies
$\strategy^1$ and $\strategy^2$.

To formalize our analysis, we exploit the ability of history-dependent policies to condition on the entire trajectory.
This yields the following structural property:

\begin{lemma}[Policy-Independent Continuation After Disagreement]
\label{lem:EqualValueFunctions}
Consider an HSG and let $\history_h \in \History$ be a reachable history containing a disagreement, i.e., $\action_j^1 \neq \action_j^2$ for some $j<h$.
Then, for any history-dependent policies $\strategy^1,\strategy^2,{\strategy^1}',{\strategy^2}'$, the continuation value is identical:
\[
\Value^{\strategy^1,\strategy^2}(\history_h)
=
\Value^{{\strategy^1}',{\strategy^2}'}(\history_h).
\]
\end{lemma}

\begin{proof}
Fix a reachable history $\history_h$ containing a disagreement.
Choose two history-dependent policies $\strategy^1,\strategy^2$
that deterministically play the actions specified by
$\history_h$ at every time step $j<h$, so that
$\p^{\strategy^1,\strategy^2}(\history_h)>0$.
Since the history contains a disagreement, swapping these
policies makes $\history_h$ unreachable:
\[
\p^{\strategy^2,\strategy^1}(\history_h)=0.
\]
Now consider another pair of policies
${\strategy^1}',{\strategy^2}'$ that agrees with the
original pair outside the continuation of $\history_h$,
but follows arbitrary continuation policies after $\history_h$.
Applying the HSG assumption to both policy pairs and
subtracting the resulting equalities yields
\[
\p^{\strategy^1,\strategy^2}(\history_h)\left(
\Value^{\strategy^1,\strategy^2}(\history_h)
-
\Value^{{\strategy^1}',{\strategy^2}'}(\history_h)
\right)=0.
\]
Since $\p^{\strategy^1,\strategy^2}(\history_h)>0$, the continuation values are equal.
As the continuation policies were arbitrary, the result follows.
\end{proof}

Combining this structural property with the EPR transformation establishes \cref{prop:hsg}.

\begin{proof}[Proof of \cref{prop:hsg}]
Since every HSG is also an MSG, by \cref{lem:MSG_is_SSG}
the game can be transformed into a value-equivalent SSG.
As the transformation preserves the value for all policies,
the transformed game also satisfies HSG.

Consider any reachable history $\history_h$ containing a
disagreement. By \cref{lem:EqualValueFunctions}, its
continuation value is independent of future policies.
Since the transformed game is an SSG, playing identical
actions at every subsequent step yields zero payoff.
Therefore, the continuation value must be zero.

In particular, choosing arbitrary actions
$\action^1,\action^2$ at $\history_h$ and identical actions
thereafter gives
\[
\hat{\utility}(\state_h,\action^1,\action^2)=0.
\]
Hence, all payoffs following the first disagreement vanish,
and the transformed game can terminate at that point
without changing its expected value.
\end{proof}

\section{Limitations}
We adopt a finite-horizon, episodic formulation to provide a clear unit of interaction and measure learning progress across repeated episodes. The finite horizon also allows us to present Expected Payoff Redistribution through backward induction, making the structural reduction and its analysis explicit.

This formulation does not exclude cyclic dynamics. A physical state may be revisited within an episode; our time-augmented representation assigns each visit to a distinct state–time pair. The resulting partitioned state space is acyclic because time advances, even when the underlying physical process contains cycles.

Our choice is primarily one of presentation and analytical scope. In finite-state discounted infinite-horizon games, backward induction can be replaced by a convergent potential transformation yielding an equivalent per-state symmetric game. Since the learner observes only states and actions, this transformation is needed only for the analysis. The finite-horizon formulation lets us state the learning guarantees without introducing discount-dependent bounds or a separate performance criterion for continuing interaction.

\section{Acknowledgments}

 This project is supported by the Israel Science Foundation (ISF, grant number 2250/22).

\bibliography{paper}
\bibliographystyle{plainnat}


\appendix

\section{Copycat in matrix games}
\label{app:Matrix-Game}

\begin{algorithm}[H]
   \caption{Copycat in matrix games}
   \label{alg:matrixgames}
\begin{algorithmic}
    \STATE Get Learning rate $\mu$
    \STATE Initialize $\utility_1(\action^1,\action^2)=0 \;\ \forall \action^1,\action^2 \in \Action$
    \FOR{$\timestep=1$ {\bfseries to} $\Time$}
        \STATE Compute: $\strategy_\timestep^1 = \argmax_{\strategy^1} \min_{\strategy^2} \mathbb{E}^{\strategy^1,\strategy^2}[\utility_{\timestep}(\strategy^1,\strategy^2)]$.
        \STATE Play $\action_\timestep^1 \sim \strategy_\timestep^1$. Observe opponent's action $\action_\timestep^2$.
        \STATE Update:  $\utility_{\timestep+1}(\action^1,\action^2) = \begin{cases}
            \utility_{\timestep}(\action^1,\action^2) - \frac{\mu}{2}, & \action^1=\action_\timestep^1,\action^2=\action_\timestep^2, \action^1_\timestep\neq \action^2_\timestep; \\
            \utility_{\timestep}(\action^1,\action^2) + \frac{\mu}{2}, & \action^1=\action_\timestep^2,\action^2=\action_\timestep^1, \action^1_\timestep\neq \action^2_\timestep; \\
            \utility_{\timestep}(\action^1,\action^2), & \text{otherwise}
            \end{cases}$
        \STATE Clip $\utility_{\timestep+1}$ to $[-1,1]$ 
    \ENDFOR
\end{algorithmic}
\end{algorithm}

In this section we will introduce an alternative approach to analyzing the copycat problem in symmetric zero-sum matrix games via a reduction to Online Convex Optimization.
To that end, we view the algorithm as introducing a third virtual player who observes the actions of the other two players at each round. 
At time  $\timestep$, this virtual player selects $\utility_\timestep \in \Utility$, where $\Utility$ represents the set of all skew-symmetric matrices of size $\actionsize\times\actionsize$.
The virtual player's loss, given that the original players chose actions $\action^1$ and $\action^2$ is $\utility_\timestep(\action^1,\action^2)$.
Like the learner, this virtual player does not observe the true payoff of the game, and updates their plays using only the past plays of the learner and the opponent. 
The goal of the virtual player is to decrease the virtual payoff in action pairs $\action^1,\action^2$ that were played in the past, in order to encourage the players to choose actions $\action^2,\action^1$ in the next rounds, which will motivate symmetric play over time. 

The virtual player's strategy relies on an online Gradient Descent algorithm over the set of skew-symmetric bounded matrices, defined as
\[
    \Utility =  \{\utility \in \Real^{\actionsize \times \actionsize} : |\utility(\action^1,\action^2)| \le 1 ~\text{and}~ \utility(\action^1,\action^2)=-\utility(\action^2,\action^1), \; \forall \action^1,\action^2 \in \Action\}.
\]

The loss matrices received by $\ogd$ in $\Time$ rounds are $\loss_1,\cdots,\loss_\Time \in \mathbb{R}^{\actionsize \times \actionsize}$ where each $\loss_\timestep$ is a matrix whose entries are 0 except for a single entry $\action^1_\timestep, \action^2_\timestep$, corresponding to the actions chosen by both players on round $\timestep$, that equals 1.
Note that $\|\utility\|_F \le \actionsize$ for all $\utility \in \Utility$ and that $\|\loss_\timestep\|_F \le 1$ for all $\timestep=1,\dots,\Time$.
As an example for such an algorithm, we can take Online Gradient Descent that has the following regret guarantee (see, e.g., \citealp{hazan2022introduction}):
\begin{equation}
    \label{eq:ogd-matrix-regret}
    \sum_{\timestep=1}^\Time \utility_\timestep \cdot \loss_\timestep - \min_{\utility \in \Utility} \sum_{\timestep=1}^\Time \utility \cdot \loss_\timestep
    =
    O(\actionsize \sqrt{\Time}).
\end{equation}

Given such an online Gradient Descent algorithm, we introduce the learner's strategy in \cref{alg:matrixgames}.
On each round $\timestep$, $\ogd$ outputs a skew-symmetric payoff function $\utility_\timestep \in \Utility$. 
The learner then selects a safety-level strategy $\strategy_\timestep^1$ based on $\utility_\timestep$. 
This approach guarantees a non-negative expected payoff w.r.t. $\utility_\timestep$, ensuring that:
\begin{align}
    \label{eq:safety-level-bound}
    \mathbb{E}\brk[s]
    *{\utility_\timestep(\action^1_\timestep,\action^2_\timestep)} \geq 0, \quad \forall t=1,\dots,T,
\end{align}
regardless of the strategies $\strategy_1^2,\dots,\strategy_\Time^2$ chosen by the adversary.

Now indeed suppose that the virtual player uses Online Gradient Descent, let $\utility_\star$ represent the actual payoffs of the game. We bound the learner's payoff as follows:

\begin{align*}
    \E \brk[s]*{\abs*{\sum_{\timestep=1}^\Time\utility_\star(\action^1_\timestep,\action^2_\timestep)}}
    & = \E \brk[s]*{-\min_{\utility \in \{\utility_\star, -\utility_\star\}}\sum_{\timestep=1}^\Time\utility(\action^1_\timestep,\action^2_\timestep)}
    \tag{$\utility$ is symmetric}\\
    &\leq
    \E \brk[s]*{-\min_{\utility \in \Utility} \sum_{\timestep=1}^\Time\utility(\action^1_\timestep,\action^2_\timestep)} 
    \tag{$\utility_\star, -\utility_\star \in \Utility$} \\
    &\leq
    \E \brk[s]*{\sum_{\timestep=1}^\Time \utility_\timestep(\action_\timestep^1,\action_\timestep^2)-
    \min_{\utility \in \Utility} \sum_{\timestep=1}^\Time\utility(\action^1_\timestep,\action^2_\timestep)}
    \tag{\cref{eq:safety-level-bound}} \\
    &=
    O(\actionsize \sqrt{\Time}).
    \tag{\cref{eq:ogd-matrix-regret}}
\end{align*}

We have shown this method achieves a bound comparable to that of \cite{feldman2010playing}. 
To complete the analysis, we now show \cref{alg:matrixgames} runs in polynomial time.

\subsection{Computational Complexity}
\label{app:Matrix-game-complexity}
At each round, \cref{alg:matrixgames} begins with a skew-symmetric matrix $\utility$ (for example the zero matrix), and updates only a single entry $(\action_\timestep^1, \action_\timestep^2)$ of the matrix. 
Consequently, the algorithm must project the updated matrix onto the set of all skew-symmetric matrices. 
Due to the linear nature of the projection and the fact that only one entry is modified, it suffices to project the updated entry $(\action_\timestep^1, \action_\timestep^2)$ along with its corresponding symmetric counterpart $(\action_\timestep^2, \action_\timestep^1)$, which takes $O(1)$ time.
As computing a safety-level strategy can be done in $\poly(\actionsize)$ time, \cref{alg:matrixgames} runs in polynomial time in its entirety.

\section{Lower Bound Analysis}
\label{app:lower_bound}
Since the learner does not observe payoffs, they are compelled to play a strategy that is robust against any potential payoff function.
This setting is analytically equivalent to a scenario where the adversary selects the worst-case payoff function after the game concludes.
We begin by analyzing the fundamental case of a repeated matrix game.

Let $\actioncount_\Time(\action^1,\action^2)$ denote the number of rounds in which the learner played action $\action^1$ and the adversary played action $\action^2$. We define the discrepancy $\actiondelta_\Time(\action^1,\action^2):=\actioncount_\Time(\action^2,\action^1)-\actioncount_\Time(\action^1,\action^2)$.

Consider the following adversarial payoff function, defined in hindsight based on the final action discrepancies:
\[
\utility(\action^1,\action^2)=
\begin{cases}
    1, & \actiondelta_\Time(\action^2,\action^1) < 0;\\
    -1, & \actiondelta_\Time(\action^2,\action^1)>0; \\
    0, &\text{otherwise}
\end{cases}
\]
Observe that for this specific payoff function, the cumulative contribution of any symmetric pair $(i,j)$ satisfies:
\[
\utility(\action^1,\action^2)\actioncount(\action^1,\action^2)+\utility(\action^2,\action^1)\actioncount(\action^2,\action^1) = -\abs*{\actiondelta_\Time(\action^1,\action^2)}
\]
Substituting this into the expected payoff guarantee yields:
\begin{align}
\label{eq:payoff-to-difference}
\E\brk[s]*{\abs*{\sum_{\timestep=1}^\Time\utility(\action^1_\timestep,\action^2_\timestep)}}
&=\E\brk[s]*{\abs*{\sum_{\action^1<\action^2}\utility(\action^1,\action^2)\actioncount(\action^1,\action^2) + \utility(\action^2,\action^1)\actioncount(\action^2,\action^1)}}\\
&=\E\brk[s]*{\abs*{\sum_{\action^1<\action^2}-\abs*{\actiondelta_\Time(\action^1,\action^2)}}} \nonumber\\
&=\E\brk[s]*{\sum_{\action^1<\action^2}\abs*{\actiondelta_\Time(\action^1,\action^2)}}\nonumber
\end{align}

Thus, to lower bound the payoff, it suffices to lower bound the expected total discrepancy $\E\brk[s]*{\sum_{\action^1<\action^2}\abs*{\actiondelta_\Time(\action^1,\action^2)}}$

\subsection{Matrix game}
We establish a lower bound of $\Omega(\actionsize\sqrt{\Time})$ on the expected cumulative regret.
We assume that the number of actions is fixed and satisfies $\actionsize > 2$, that the time horizon $\Time$ is sufficiently large, and that the adversary plays a \textbf{uniform strategy}, choosing each action with probability $1/\actionsize$ independently at every step.

Let $\actioncount_\timestep(\action)$ denote the number of times the learner plays action $\action$ up to time $\timestep$. We analyze two regimes based on the learner's action frequencies.

\paragraph{Case 1: The Unbalanced Regime}

Suppose there exists an action $\action^1$ such that the learner's expected count deviates significantly from uniformity:
\begin{equation*}
    \E\brk[s]*{\abs*{\actioncount_\Time(\action^1) - \frac{\Time}{\actionsize}}} \ge 2\actionsize\sqrt{\Time}.
\end{equation*}

To lower bound the regret, we decompose the discrepancy using the Reverse Triangle Inequality:
\begin{align*}
    \E\brk[s]*{\abs*{\actiondelta_\Time(\action^2,\action^1)}} &\ge \underbrace{\frac{1}{\actionsize} \E\brk[s]*{\abs*{\actioncount_\Time(\action^1) - \actioncount_\Time(\action^2)}}}_{\text{Signal}} \\
    &\qquad\qquad- \underbrace{ \E\brk[s]*{\abs*{\actioncount_\Time(\action^1,\action^2) - \actioncount_\Time(\action^2,\action^1) - \frac{1}{\actionsize}\actioncount_\Time(\action^1) + \frac{1}{\actionsize}\actioncount_\Time(\action^2)}}}_{\text{Noise}}.
    \label{eq:lower_bound_decomp}
\end{align*}

First, we bound the \textbf{Total Signal}. Using the inequality $\sum_{\action^2\in\Action} |\actioncount_\Time(\action^1) - \actioncount_\Time(\action^2)| \ge \actionsize |\actioncount_\Time(\action^1) - \frac{\Time}{\actionsize}|$, we obtain:
\begin{equation}
    \sum_{\action^2 \ne \action^1} \frac{1}{\actionsize} \E\brk[s]*{\abs*{\actioncount_\Time(\action^1) - \actioncount_\Time(\action^2)}} \ge \E\brk[s]*{\abs*{\actioncount_\Time(\action^1) - \frac{\Time}{\actionsize}}} \ge 2\actionsize\sqrt{\Time}.
\end{equation}

Next, we bound the \textbf{Total Noise}. We define the noise term $\mathcal{E}_{\timestep}$ as: 
\[\mathcal{E}_{\timestep} \coloneqq \Ind\{\action^1_\timestep=\action^1,\action^2_\timestep=\action^2\} - \Ind(\action^1_\timestep=\action^2,\action^2_\timestep=\action^1) - \frac{1}{\actionsize}\Ind(\action^1_\timestep=\action^1) + \frac{1}{\actionsize}\Ind(\action^1_\timestep=\action^2).
\]

The total noise corresponds to the magnitude of the sum of these variables. By applying Jensen's inequality and bounding the variance, we have:
\[
\E\brk[s]*{\abs*{\sum_{\timestep=1}^\Time \mathcal{E}_\timestep}}
\;\le\;
\sqrt{\Var\brk*{\sum_{\timestep=1}^\Time \mathcal{E}_\timestep}}
\;\le\;
\sqrt{\frac{1}{\actionsize}\E\brk[s]*{\actioncount_\Time(\action^1)+\actioncount_\Time(\action^2)}}.
\]

Summing this bound over all $\action^2\neq \action^1$ and applying the Cauchy--Schwarz inequality to the sum of $\actionsize-1$ terms yields:
\[
\sum_{\action^2\neq \action^1} \E\brk[s]*{\abs*{\sum_{\timestep=1}^\Time \mathcal{E}_\timestep}}\le
\sqrt{(\actionsize-1)\sum_{\action^2\neq \action^1}\frac{1}{\actionsize}\E[\actioncount_\Time(\action^1)+\actioncount_\Time(\action^2)]}
\le
\sqrt{\actionsize\Time}.
\]

\textbf{Conclusion.}
Combining the signal bound with the noise bound above yields
\[
\E\brk[s]*{\sum_{\action^1<\action^2}\abs*{\actiondelta_\Time(\action^2,\action^1)}}\ge\sum_{\action^2\neq \action^1}\E\brk[s]*{\abs*{\actiondelta_\Time(\action^2,\action^1)}}
\ge
2n\sqrt{T}-\sqrt{nT}
\ge
n\sqrt{T},
\]
which holds for all sufficiently large $\Time$.

\paragraph{Case 2: The Balanced Regime}
Suppose that for all actions $\action\in\Action$, the counts are concentrated near the uniform allocation:
\[
\E\brk[s]*{\abs*{\actioncount_\Time(\action) - \frac{\Time}{\actionsize}}} < 2\actionsize \sqrt{\Time}.
\]

We note that by \textbf{Jensen's Inequality}:
\[
\abs*{\E\brk[s]*{\actioncount_\Time(\action)} - \frac{\Time}{\actionsize}} \le \E\brk[s]*{\abs*{\actioncount_\Time(\action) - \frac{\Time}{\actionsize}}}< 2\actionsize \sqrt{\Time}.
\]

We lower-bound using the specific payoff of a Rock-Paper-Scissors (RPS) game on the cycle $\Delta = \{\action^1, \action^2, \action^3\}$. Define the payoff function $\utility(\action^1_\timestep, \action^2_\timestep)$ as $1$ if the pair $(\action^1_\timestep, \action^2_\timestep)$ is a winning edge for the learner in $\Delta$, $-1$ if it is a losing edge, and $0$ otherwise.

Let $X \coloneqq \sum_{\timestep=1}^\Time \utility(\action^1_\timestep, \action^2_\timestep)$. By the triangle inequality, the total discrepancy on the cycle upper-bounds the magnitude of this cumulative payoff:
\[
\sum_{\substack{\action^1<\action^2\\
\action^1,\action^2\in\Delta}} \abs*{\actiondelta_\Time(\action^1,\action^2)} \ge \abs*{X}.
\]
To lower-bound $\E[\abs{X}]$, we employ the Fourth Moment Method (a consequence of Hölder's inequality):
\[
\E\brk[s]{\abs{X}} \ge \frac{(\E\brk[s]{X^2})^{3/2}}{(\E\brk[s]{X^4})^{1/2}}.
\]

\textbf{The Second Moment.}
The sequence $(\utility(\action^1_\timestep, \action^2_\timestep))_{\timestep=1}^\Time$ forms a martingale difference sequence. The conditional variance at step $\timestep$ is non-zero only if the learner plays an action $i_\timestep \in \Delta$ and the adversary (playing uniformly) hits a neighbor in the cycle (with probability $2/\actionsize$). Summing over all time steps:
\[
\E[X^2] 
\;=\; 
\sum_{t=1}^\Time \E\brk[s]*{\E\brk[s]*{\utility(\action^1_\timestep, \action^2_\timestep)^2 \mid \action^1_1,\action^2_1\dots \action^1_{\timestep-1},\action^2_{\timestep-1}, \action^1_\timestep}} 
= 
\sum_{\timestep=1}^\Time \E\brk[s]*{\frac{2}{\actionsize} \Ind_{\{\action^1_\timestep \in \Delta\}}} 
= 
\frac{2}{\actionsize} \E\brk[s]*{\actioncount_\Delta},
\]
where $\actioncount_\Delta = \sum_{\action\in\Delta}\actioncount_\Time(\action)$ denotes the total number of times the learner plays in the cycle.

\textbf{The Fourth Moment.}
Let $Y_\timestep=\utility(\action^1_\timestep,\action^2_\timestep)$ and
$q_\timestep=\frac{2}{\actionsize}
\mathbf{1}\{\action^1_\timestep\in\Delta\}$.
Since the opponent plays uniformly,
$\E[Y_\timestep^2\mid\mathcal F_{\timestep-1},\action^1_\timestep]
=q_\timestep$.
Thus, $Z_\timestep=Y_\timestep^2-q_\timestep$ is a
martingale difference sequence, and
\[
\E\brk[s]*{\brk*{\sum_{\timestep=1}^\Time
Z_\timestep}^2}
\le \sum_{\timestep=1}^\Time\E[q_\timestep]
=\frac{2}{\actionsize}\E\brk[s]*{\actioncount_\Delta}.
\]

By the Burkholder--Davis--Gundy inequality,
\begin{align*}
\E\brk[s]*{X^4}
&\le C\E\brk[s]*{
\brk*{\sum_{\timestep=1}^\Time Y_\timestep^2}^2}\\
&=C\E\brk[s]*{
\brk*{\sum_{\timestep=1}^\Time
(Z_\timestep+q_\timestep)}^2}\\
&\le 2C\E\brk[s]*{
\brk*{\sum_{\timestep=1}^\Time Z_\timestep}^2}
+2C\E\brk[s]*{
\brk*{\sum_{\timestep=1}^\Time q_\timestep}^2}\\
&\le \frac{4C}{\actionsize}
\E\brk[s]*{\actioncount_\Delta}
+\frac{8C}{\actionsize^2}
\E\brk[s]*{\actioncount_\Delta^2}.
\end{align*}

(Note: The last inequality follows from the martingale
difference property of $Z_\timestep$, together with
$(a+b)^2\leq 2a^2+2b^2$.)


Since $\actioncount_\Delta$ scales linearly with $\Time$ (specifically $\E[\actioncount_\Delta] \propto \Time/\actionsize$), the quadratic term $\E[\actioncount_\Delta^2]$ grows as $\Time^2$, dominating the linear term $\E[\actioncount_\Delta]$. Thus, for sufficiently large $\Time$, we can absorb the linear term into the quadratic term:
\[
\E\brk[s]*{X^4} \le \frac{16C}{\actionsize^2} \; \E\brk[s]*{\actioncount_\Delta^2}.
\]

\textbf{Bounding the Second Moment.}
We show that the second moment $\E[\actioncount_\Delta^2]$ is dominated by the squared mean.
Since $\actioncount_\Delta$ and $\E\brk[s]*{\actioncount_\Delta}$ are both restricted to the interval $[0, \Time]$, their distance is at most $\Time$. Thus, almost surely:
\[
(\actioncount_\Delta - \E\brk[s]*{\actioncount_\Delta})^2 \le \Time \abs{\actioncount_\Delta - \E\brk[s]*{\actioncount_\Delta}}.
\]
Taking expectations and applying the regime condition $\E\brk[s]{\abs{\actioncount(\action) - \Time/\actionsize} < 2\actionsize\sqrt{\Time}}$:
\begin{align*}
\Var[\actioncount_\Delta] &= \E\brk[s]*{(\actioncount_\Delta - \E\brk[s]*{\actioncount_\Delta})^2} \le \Time \sum_{\action \in \Delta} \E\brk[s]*{\abs{\actioncount_\action - \E\brk[s]*{\actioncount_\action}}} \\
& \leq \sum_{\action \in \Delta}2\Time\E\brk[s]*{\abs*{\actioncount_\Time(\action)-\frac{\Time}{\actionsize}}}\leq 12\actionsize\Time\sqrt{\Time}
\end{align*}

The variance scales as $O(\Time^{1.5})$, which is asymptotically negligible compared to $\E\brk[s]*{\actioncount_\Delta}^2 \propto \Time^2$.
Implying that for large enough $\Time>\actionsize^6$ we will have 
\[
\E[\actioncount_\Delta^2]=\Var[\actioncount_\Delta] + \E[\actioncount_\Delta]^2 \leq 3\E[\actioncount_\Delta]^2
\]

\textbf{Completing the Bound.}
We now substitute these moments into the Hölder's ratio. Squaring both sides implies:
\[
\E\brk[s]{\abs{X}}^2 \ge \frac{\brk*{\E\brk[s]{X^2}}^3}{\E\brk[s]{X^4}} \geq \frac{(\frac{2}{\actionsize} \E\brk[s]*{\actioncount_\Delta})^3}{\frac{16C}{\actionsize^2} \; \E\brk[s]*{\actioncount_\Delta^2}}.
\]

Substituting the second moment bound into the denominator:
\[
\E\brk[s]{\abs{X}}^2\ge \frac{(\frac{2}{\actionsize} \E\brk[s]*{\actioncount_\Delta})^3}{\frac{16C}{\actionsize^2} \cdot3 \; \E\brk[s]*{\actioncount_\Delta}^2} \ge \frac{1}{6C\actionsize}\E\brk[s]*{\actioncount_\Delta} \ge \frac{1}{2C\actionsize}(\frac{\Time}{\actionsize}-2\actionsize\sqrt{\Time}) = \Omega\left(\frac{\Time}{\actionsize^2}\right).
\]
Finally, taking the square root implies:
\[
\E\brk[s]{\abs{X}} = \Omega\left(\frac{\sqrt{\Time}}{\actionsize}\right).
\]

\textbf{Aggregation.}
Summing over all $\binom{\actionsize}{3}$ cycles and dividing by the overcounting factor $(\actionsize-2)$ (since each edge is part of $\actionsize-2$ triangles):
\[
\E\brk[s]*{\sum_{i<j}\abs*{\actiondelta_\Time(i,j)}}\;\ge\; \frac{1}{\actionsize-2} \sum_{\{i,j,k\}} \E\brk[s]*{\abs*{X_{i,j,k}}} 
\;=\; \frac{1}{\actionsize-2} \cdot \binom{\actionsize}{3}\Omega(\frac{\sqrt{\Time}}{\actionsize})   
\;=\; \Omega(\actionsize\sqrt{\Time}).
\]

\subsection{SSG games}
We extend the lower bound to the Markovian setting by constructing a specific SSG instance with a constant $|\State_h|$ states for any step $\horizon$ and episodes of length $\Horizon$ and a starting state $\state_1$ with a zero payoff matrix.
Consider a game where the transition dynamics are uniform and independent of the current state or actions. Specifically, for any state $\state$, step $\horizon$, and action pair $(\action^1, \action^2)$:
\[
\transition(\state'_{\horizon+1}\mid\state_\horizon,\action^1,\action^2)=\frac{1}{|\State_{\horizon+1}|}
\]

In this construction, the expected visitation frequency of any state is a random variable determined solely by the transition, independent of the players' strategies.
Let $\Time_\state$ denote the random number of times state $\state$ is visited over the entire game horizon of $\Time$ episodes.
The expected number of visits is $\E[\Time_\state] = \frac{\Time}{|\State_\horizon|}$.

Since the states are decoupled by the uniform transitions, the game effectively decomposes into $|\State|$ independent repeated matrix games, where the effective time horizon for state $\state$ is $\Time_\state$.
Applying the reduction from \cref{eq:payoff-to-difference} locally to each state:
\begin{align*}
\E\brk[s]*{\abs*{\sum_{\timestep=1}^\Time\sum_{\horizon=1}^\Horizon\utility(\state_{\horizon,\timestep},\action^1_\timestep,\action^2_\timestep)}} 
&= \E\brk[s]*{\sum_{\horizon=2}^\Horizon\sum_{\state\in\State_\horizon}\sum_{\action^1<\action^2}\abs*{\actiondelta_{\Time_\state}(\state,\action^1,\action^2)}}
\end{align*}

We assume that $\Time \ge C\actionsize^6|\State|$
for a sufficiently large constant $C$.
For each $\state\in\State_\horizon$, $\horizon\ge2$,
we have $\Time_\state\sim
\operatorname{Binomial}(\Time,1/|\State_\horizon|)$.
By a Chernoff bound,
\[
\Pr\left(
\Time_\state\ge\frac{\Time}{2|\State_\horizon|}
\right)\ge\frac12.
\]
Conditioning on the visitation sequence and applying
the matrix-game lower bound yields, for every such state,
\[
\E\brk[s]*{
\sum_{\action^1<\action^2}
\abs*{\actiondelta_{\Time_\state}
(\state,\action^1,\action^2)}
}
=\Omega\left(
\actionsize\sqrt{\frac{\Time}{|\State_\horizon|}}
\right).
\]

\begin{align*}
\E\brk[s]*{\abs*{\sum_{\timestep=1}^\Time\sum_{\horizon=1}^\Horizon\utility(\state_{\horizon,\timestep},\action^1_\timestep,\action^2_\timestep)}} 
&=\E\brk[s]*{\sum_{\horizon=2}^\Horizon\sum_{\state\in\State_\horizon}\sum_{\action^1<\action^2}\abs*{\actiondelta_{\Time_\state}(\state,\action^1,\action^2)}}\\
&=
\sum_{\horizon=2}^\Horizon\sum_{\state\in\State_\horizon}\E\brk[s]*{\sum_{\action^1<\action^2}\abs*{\actiondelta_{\Time_\state}(\state,\action^1,\action^2)}}\\
&\ge |\State_\horizon|(\Horizon-1)\cdot\Omega\left(\actionsize\sqrt{\frac{\Time}{|\State_\horizon|}}
\right)\\
&= \Omega\left(\actionsize\sqrt{\Horizon|\State|\Time}
\right). \tag{$|\State|=1+(\Horizon-1)|\State_\horizon|$}
\end{align*}

\section{Matrix game as Blackwell’s Approachability}
\label{app:Blackwell’s-approachability-formulatio}

We first define the main concepts of approachability theory.
A Blackwell instance is a tuple $(\mathcal{X}, \mathcal{Y}, \Loss, \Utility)$ such that $\mathcal{X} \subset\Real^n$ is the action set of the learner, $\mathcal{Y} \subset \Real^m$ is the action set of the adversary. Both action sets are convex and compact. 
The set $\Utility \subset \Real^d$ is convex and closed, and the vector valued function $\Loss: \mathcal{X} \times \mathcal{Y} \mapsto \Real^d$ is bilinear.  

\begin{definition}[Approachable set]
    Given a Blackwell instance $(\mathcal{X}, \mathcal{Y}, \Loss, \Utility)$, we say that $\Utility$ is \emph{approachable} if there exists an algorithm $\alg$ that selects points $x^1,\dots,x^T \in \mathcal{X}$ such that, for any sequence $y^1, y^2, \dots y^T \in \mathcal{Y}$, we have:
    \[
        \mathrm{dist} \brk*{\frac{1}{\Time}\sum_{\timestep=1}^\Time\Loss(x^\timestep,y^\timestep),\Utility}\to 0 \quad\text{as}\quad \Time \to \infty,
    \]
where $x^\timestep \gets \alg(y^1, y^2, \dots y^{\timestep-1})$.
\end{definition}

Finally, we state Blackwell’s Approachability Theorem.
\begin{theorem}(Approachability Theorem; \citealp{blackwell1956analog})
\label{theorem:Blackwell}
    For any Blackwell instance $(\mathcal{X}, \mathcal{Y}, \Loss, \Utility)$, $\Utility$ is approachable if and only if $\;\forall y \in \mathcal{Y} \;\;\exists x \in \mathcal{X} \;: \Loss(x, y) \in \Utility$.
\end{theorem}

\subsection{Formulation of the copycat problem as Blackwell's Approachability}

We next formulate the copycat problem as a Blackwell’s Approachability problem.
We define a Blackwell instance $(\Delta\Action,\Delta\Action,\Loss,\Utility)$, such that $\Action = \{1,\dots,\actionsize\}$ is the action set of both players. 
The target set $\Utility$ is the space of all symmetric matrices: 
\[
    \Utility = \{\utility \in \Real^{\actionsize \times \actionsize} : \utility(\action^1,\action^2)=\utility(\action^2,\action^1),\; \forall \action^1,\action^2 \in \Action)\}.
\]
We define $\Loss(x,y)
= xy^\top$,
where $x,y\in\Delta\Action$. For pure actions $i,j\in\Action$, we write
$\Loss(i,j)=e_i e_j^\top$,
where $e_i$ is the $i$-th standard basis vector.

\begin{lemma}
    The copycat problem is approachable.
\end{lemma}
\begin{proof}
For each $y\in\Delta\Action$ we choose $x=y$ Then $\Loss(y,y) =yy^\top$ is symmetric
and therefore belongs to $\Utility$.
\end{proof}

Lastly, consider a repeated symmetric matrix game with payoff $\utility_\star$.
Then, since the problem is approachable:
\begin{align*}
    \abs*{\frac1\Time \mathbb{E}\brk[s]*{\sum_{t=1}^T \utility_\star(\action^1_\timestep,\action^2_\timestep)}}
    =
    \abs*{\utility_\star \cdot \mathbb{E}\brk[s]*{\frac1\Time \sum_{t=1}^T \Loss(\action^1_\timestep,\action^2_\timestep)}} 
    \to
    0.
\end{align*}

\section{Proof of \cref{lem:shorter-game}}
\label{app:shorter-game-proof}
\begin{proof}
    We look at $\strategy^1, \hat{\strategy}^1, \strategy^2, \hat{\strategy}^2 \in \Strategy_M$ under the special case where $\strategy^1 = \hat{\strategy}^1$ and $\strategy^2=\hat{\strategy}^2$ for all $\horizon$ except $\Horizon$ and for all states except $\state \in \State_\Horizon$.
    In this case:
\begin{align*}
    &\Value^{\strategy^1, \strategy^2}(\state_1) - \Value^{\hat{\strategy}^1, \hat{\strategy}^2}(\state_1)  \\
    &\qquad =
    \p^{\strategy^1, \strategy^2}[\state_{\Horizon}=\state]\brk*{\E^{\strategy^1,\strategy^2}\brk[s]*{\utility(\state, \action_\Horizon^1, \action_\Horizon^2) \mid \state_\Horizon = \state} 
    - 
    \E^{\hat{\strategy}^1,\hat{\strategy}^2}\brk[s]*{\utility(\state, \action_\Horizon^1, \action_\Horizon^2) \mid \state_\Horizon = \state}}.
\end{align*}

\cref{assum:MSG} implies:
\begin{align*}
    &\Value^{\strategy^1, \strategy^2}(\state_1) - \Value^{\hat{\strategy}^1, \hat{\strategy}^2}(\state_1) + \Value^{\strategy^2, \strategy^1}(\state_1) - \Value^{\hat{\strategy}^2, \hat{\strategy}^1}(\state_1) \\
    &\qquad = 
    \brk{\underbrace{\Value^{\strategy^1, \strategy^2}(\state_1) + \Value^{\strategy^2, \strategy^1}(\state_1)}_{=0}}
    - 
    \brk{\underbrace{\Value^{\hat{\strategy}^1, \hat{\strategy}^2}(\state_1) + \Value^{\hat{\strategy}^2, \hat{\strategy}^1}(\state_1)}_{=0}} \\
    &\qquad =0. 
\end{align*}
We combine both observations to get:
\begin{align}
\label{eq:Item-1-final}
    &\p^{\strategy^1, \strategy^2}[\state_{\Horizon}=\state](\E^{\strategy^1,\strategy^2}\brk[s]*{\utility(\state, \action_\Horizon^1, \action_\Horizon^2)\mid \state_\Horizon = \state} - \E^{\hat{\strategy}^1,\hat{\strategy}^2}\brk[s]*{\utility(\state, \action_\Horizon^1, \action_\Horizon^2)\mid \state_\Horizon = \state}) \\
    &\qquad=
    -\p^{\strategy^2, \strategy^1}[\state_{\Horizon}=\state] (\E^{\strategy^2,\strategy^1}\brk[s]*{\utility(\state, \action_\Horizon^1, \action_\Horizon^2)\mid \state_\Horizon = \state} - \E^{\hat{\strategy}^2,\hat{\strategy}^1}\brk[s]*{\utility(\state, \action_\Horizon^1, \action_\Horizon^2)\mid \state_\Horizon = \state}). 
    \nonumber
\end{align}
Now we set $\strategy^1,\strategy^2$ such that they will always choose action 1 in $\state \in \State_\Horizon$.
Formally, $
\p^{\strategy^1}[\action_\Horizon=1 \mid \state_\Horizon=\state]
=
\p^{\strategy^2}[\action_\Horizon=1\mid\state_\Horizon=\state]
=
1$.
Plugging these policies into \cref{eq:Item-1-final} gets: 
\begin{align}
\label{eq:Fina-Round-Gaps}
    &\p^{\strategy^1, \strategy^2}[\state_{\Horizon}=\state]\brk*{\utility(\state, 1, 1) - \E^{\hat{\strategy}^1,\hat{\strategy}^2}\brk[s]*{\utility(\state, \action_\Horizon^1, \action_\Horizon^2)\mid \state_\Horizon = \state}}
    =\\
    &\qquad\qquad-
    \p^{\strategy^2, \strategy^1}[\state_{\Horizon}=\state] \brk*{\utility(\state, 1, 1)-\E^{\hat{\strategy}^2,\hat{\strategy}^1}\brk[s]*{\utility(\state, \action_\Horizon^1, \action_\Horizon^2)\mid \state_\Horizon = \state}}.
    \nonumber
\end{align}
Next we choose $\strategy^1,\strategy^2$ that choose actions uniformly at random everywhere except in $\state \in \State_\Horizon$.
Formally:
$\p^{\strategy^i}[\action=\action_\horizon \mid \state_\horizon] 
= \frac{1}{|\Action|}$ for all $h \in [\Horizon]$, $\state_\horizon \in \State_\horizon \setminus \{\state\}$ and $i=1,2$. 
By the choice of our policies, we have that $\p^{\strategy^1, \strategy^2}[\state_{\Horizon}=\state]=\p^{\strategy^2, \strategy^1}[\state_{\Horizon}=\state] > 0$, from which
\[
    \utility(\state, 1, 1)+\utility(\state, 1, 1)
    =
    \E^{\hat{\strategy}^1,\hat{\strategy}^2}\brk[s]*{\utility(\state, \action_\Horizon^1, \action_\Horizon^2)\mid \state_\Horizon = \state}
    +
    \E^{\hat{\strategy}^2,\hat{\strategy}^1}\brk[s]*{\utility(\state, \action_\Horizon^1, \action_\Horizon^2)\mid \state_\Horizon = \state},
\]
Taking $\hat{\strategy}^2$ and $\hat{\strategy}^1$ to be deterministic proves \cref{it:last-round-symmetric}.

Now we rearrange \cref{eq:Fina-Round-Gaps}:
\begin{align}
            &\p^{\strategy^1,\strategy^2}[\state_\Horizon=\state]\utility(\state, \action^1, \action^2)
            + 
            \p^{\strategy^2,\strategy^1}[\state_\Horizon=\state]\utility(\state, \action^2, \action^1)
            =\\
            &\qquad\qquad
            \p^{\strategy^1,\strategy^2}[\state_\Horizon=\state] \; \utility(\state, 1, 1)
            +
            \p^{\strategy^2,\strategy^1}[\state_\Horizon=\state] \; \utility(\state, 1, 1)
            \nonumber
\end{align}
For simplicity we will define:
\[
    \Payoff(\utility,\strategy^1,\strategy^2,\horizon) 
    \coloneqq 
    \E^{\strategy^1,\strategy^2}\brk[s]*{\sum_{\horizon'=1}^{\horizon}\utility(\state_{\horizon'}, \action^1_{\horizon'}, \action^2_{\horizon'})} 
    + 
    \E^{\strategy^2,\strategy^1}\brk[s]*{\sum_{\horizon'=1}^{\horizon}\utility(\state_{\horizon'}, \action^1_{\horizon'}, \action^2_{\horizon'})}.
\]

We begin our proof by considering the payoff at step $\Horizon$ for policies $\strategy_1,\strategy_2 \in \Strategy_M$:
\begin{alignat*}{2}
        \Payoff(\utility,\strategy_1, \strategy_2, \Horizon) 
        &=
        \Payoff(\utility,\strategy_1,\strategy_2,\Horizon-1)
        &&+
        \E^{\strategy_1, \strategy_2}\brk[s]*{\utility(\state_\Horizon, \action_\Horizon^1,\action_\Horizon^2)} 
        +
        \E^{\strategy_2, \strategy_1} \brk[s]*{\utility(\state_\Horizon, \action_\Horizon^1, \action_\Horizon^2)} \\
        &=
        \Payoff(\utility, \strategy_1,\strategy_2,\Horizon-1)
        &&+
        \sum_{\state\in\State_\Horizon}\p^{\strategy^1, \strategy^2}[\state_\Horizon=\state] \E^{\strategy_1,\strategy_2}\brk[s]*{\utility(\state, \action_\Horizon^1, \action_\Horizon^2) \mid 
        \state_\Horizon = \state}\nonumber\\
        & &&+
        \sum_{\state\in\State_\Horizon} \p^{\strategy^2, \strategy^1}[\state_\Horizon=\state]\E^{\strategy_2,\strategy_1}\brk[s]*{\utility(\state, \action_\Horizon^1, \action_\Horizon^2)\mid 
        \state_\Horizon = \state}\nonumber\\
        &=
        \Payoff(\utility,\strategy_1,\strategy_2,\Horizon-1) 
        &&+
        \sum_{\state\in\State_\Horizon}\p^{\strategy^1,\strategy^2}[\state_\Horizon=\state]\utility(\state,1,1) \nonumber \\
        & &&+ 
        \sum_{\state\in\State_\Horizon}\p^{\strategy^2,\strategy^1}[\state_\Horizon=\state]\utility(\state,1,1).
        \tag{\cref{eq:Fina-Round-Gaps}}\\
\end{alignat*}
Expanding the last two summands using the Markov property:
\begin{align*}
    &\sum_{\state\in\State_\Horizon}\p^{\strategy_1,\strategy_2}[\state_\Horizon=\state]\utility(\state,1,1) + \sum_{\state\in\State_\Horizon} \p^{\strategy_2,\strategy_1}[\state_\Horizon=\state]\utility(\state,1,1) \\
    &\qquad\qquad=
    \E^{\strategy_1,\strategy_2}\brk[s]*{\sum_{\state\in\State}\transition(\state\mid\state_{\Horizon-1}, \action^1_{\Horizon-1},\action^2_{\Horizon-1})\utility(\state,1,1)}
    \\
    &\qquad\qquad\qquad\qquad+ 
    \E^{\strategy_2,\strategy_1}\brk[s]*{\sum_{\state\in\State}\transition(\state\mid\state_{\Horizon-1}, \action^1_{\Horizon-1},\action^2_{\Horizon-1})\utility(\state,1,1)}.
\end{align*}
Finally, substituting this back into the original expression for the payoff:
\[
    \Payoff(\utility,\strategy_1, \strategy_2, \Horizon) =\Payoff(\hat\utility,\strategy_1,\strategy_2,\Horizon-1),
\]
We can clearly see that the left side must satisfy \cref{assum:MSG}, hence the right side must as well, completing our proof.

\end{proof}

\section{Strategy: Learn on Each State Independently}
\label{sec:online-proof}

At first glance, generalizing from matrix games to Markov games suggests that the adversary might exploit the transition dynamics to gain an advantage. 
Since players' actions influence future state visitations, one might assume a holistic approach accounting for long-term planning is required. 
However, we demonstrate that a local Copycat strategy is surprisingly robust to unknown dynamics. We show that for SSG, the learner can effectively decouple the transition dynamics from the immediate payoffs. By simply treating each state as an independent matrix game, the learner can ignore the transition function $\transition$ entirely while still guaranteeing sublinear regret, avoiding any exponential scaling with the game's parameters.

The learner's strategy is to select a safety-level strategy $\strategy_\timestep^1$ at each state, treating $\utility_\timestep$ as the payoff function. 
As $\utility_\timestep(\state,\cdot,\cdot)$ is skew-symmetric, the value of the game is zero, guaranteeing a non-negative expected payoff for the learner:
\begin{align}
    \label{eq:markov-safety-level-bound}
    \mathbb{E}^{\strategy_\timestep^1,\strategy_\timestep^2}\brk[s]
    *{\sum_{\horizon=1}^\Horizon \utility_\timestep(\state_{\horizon,\timestep}, \action^1_{\horizon,\timestep},\action^2_{\horizon,\timestep})} \geq \sum_{\horizon=1}^\Horizon\mathbb{E}^{\strategy_\timestep^1,\strategy_\timestep^2}\brk[s]
    *{\utility_\timestep(\state_{\horizon,\timestep}, \action^1_{\horizon,\timestep},\action^2_{\horizon,\timestep})} \geq 0,
\end{align}
regardless of the strategy $\strategy_\timestep^2$ chosen by the adversary.

Note that $\|\utility\|_F \le \actionsize\sqrt{|\State|}$ for all $\utility \in \Utility$ and that $\|\loss_\timestep\|_F \le \sqrt{\Horizon}$ for all $\timestep=1,\dots,\Time$.
Hence OGD yields:
\begin{equation}
    \label{eq:ogd-SSG-regret}
    \sum_{\timestep=1}^\Time \utility_\timestep \cdot \loss_\timestep
    -
    \min_{\utility \in \Utility} \sum_{\timestep=1}^\Time \utility \cdot \loss_\timestep 
    =
    O(\actionsize \sqrt{\Time|\State|\Horizon}).
\end{equation}

Now indeed suppose that the virtual player uses OGD.
Let $\utility_\star$ represent the actual payoffs of the game. 
We bound the learner's playoff as follows:
\begin{align*}
    \E \brk[s]*{\abs*{\sum_{\timestep=1}^\Time\sum_{\horizon=1}^\Horizon\utility(\state_{\horizon,\timestep}, \action_{\horizon,\timestep}^1,\action_{\horizon,\timestep}^2)}}
    &= \E\brk[s]*{-\min_{\utility\in\{\utility_\star,-\utility_\star\}}\sum_{\timestep=1}^\Time\sum_{\horizon=1}^\Horizon\utility(\state_{\horizon,\timestep}, \action_{\horizon,\timestep}^1,\action_{\horizon,\timestep}^2)} \tag{\cref{assum:SSG}}\\
    &\leq
     \E \brk[s]*{-\min_{\utility \in \Utility}\sum_{\timestep=1}^\Time\sum_{\horizon=1}^\Horizon\utility(\state_{\horizon,\timestep},\action^1_{\horizon,\timestep},\action^2_{\horizon,\timestep})} 
    \tag{$\utility_\star, -\utility_\star \in \Utility$} \\
    &\leq
    \E \Bigg[\sum_{\timestep=1}^\Time\sum_{\horizon=1}^\Horizon \utility_\timestep(\state_{\horizon,\timestep},\action_{\horizon,\timestep}^1,\action_{\horizon,\timestep}^2)\\
    &\qquad-
    \min_{\utility \in \Utility} \sum_{\timestep=1}^\Time\sum_{\horizon=1}^\Horizon\utility(\state_{\horizon,\timestep},\action^1_{\horizon,\timestep},\action^2_{\horizon,\timestep})
    \Bigg]
    \tag{\cref{eq:markov-safety-level-bound}} \\
    &=
    O(\actionsize \sqrt{\Time|\State|\Horizon}).
    \tag{\cref{eq:ogd-SSG-regret}}
\end{align*}

Furthermore, we claim that \cref{alg:Main}  runs in polynomial time. 
Indeed, computing a safety-level strategy takes polynomial time \citep{wal1976markov}, and the rest of the steps of the algorithm also run in polynomial time.
It remains to show that OGD over $\Utility$ runs in polynomial time, for which it suffices to explain how to perform efficient Euclidean projection onto the set $\Utility$ (see \cref{sec:olo}).
As the set $\Utility$ is defined such that the payoff matrix in each state is skew-symmetric, projection on this set can be done by projecting the payoff in each state onto the space of skew-symmetric matrices separately, which runs in linear time (see \cref{app:Matrix-game-complexity}).
Thus projections on $\Utility$ take $O(|\State|\actionsize^2)$ time in total.
We note that the approach presented here serves as a more structural formulation of the basic Copycat algorithm. Furthermore, we show in \cref{sec:copycat_on_each_step} that the standard Copycat algorithm can also be generalized to solve SSG environments with similar theoretical guarantees.

\paragraph{Connection to cumulative policy values.}
For any fixed payoff function $u$, the difference between
the realized return in episode $t$ and
$V^{\pi_t^1,\pi_t^2}(s_1)$ has conditional expectation zero,
given the preceding history and the policies selected for
that episode. These differences are therefore martingale
differences. Since each episode's return lies in $[-H,H]$,
its conditional variance is at most $H^2$. Thus, by
Cauchy--Schwarz and the orthogonality of martingale differences,
\begin{align*}
&\E\brk[s]*{\abs*{\sum_{\timestep=1}^{\Time}\brk*{\sum_{\horizon=1}^{\Horizon}\utility(\state_{\horizon,\timestep},\action^1_{\horizon,\timestep},\action^2_{\horizon,\timestep})-
\Value^{\strategy_\timestep^1,\strategy_\timestep^2}(\state_1)}}}\\
&\qquad\le
\sqrt{\sum_{\timestep=1}^{\Time}\E\brk[s]*{\brk*{\sum_{\horizon=1}^{\Horizon}\utility(\state_{\horizon,\timestep},\action^1_{\horizon,\timestep},\action^2_{\horizon,\timestep})-
\Value^{\strategy_\timestep^1,\strategy_\timestep^2}(\state_1)}^2}}\le H\sqrt{T}.
\end{align*}
Consequently, the triangle inequality yields
\[
\E \abs*{\sum_{\timestep=1}^{\Time}\Value^{\strategy_\timestep^1,\strategy_\timestep^2}(\state_1)
}
\le
\E \abs*{\sum_{\timestep=1}^{\Time}\sum_{\horizon=1}^{\Horizon}
\utility(\state_{\horizon,\timestep},\action^1_{\horizon,\timestep},\action^2_{\horizon,\timestep})
}
+
\Horizon\sqrt{\Time}.
\]
Since $|\State|\ge \Horizon$ in the layered state space, this preserves
the $O(\actionsize\sqrt{\Time|\State|\Horizon})$ upper bound.

\section{Strategy: Copycat on Each State}
\label{sec:copycat_on_each_step}

At first glance, generalizing from matrix games to Markov games suggests that the adversary might exploit the transition dynamics to gain an advantage. 
Since players' actions influence future state visitations, one might assume a holistic approach accounting for long-term planning is required. 
However, we demonstrate that the Copycat strategy is surprisingly robust to unknown dynamics. We show that for SSG, the learner can effectively decouple the transition dynamics from the immediate payoffs. By simply treating each state as an independent matrix game and applying the Copycat strategy locally, the learner can ignore the transition function $\transition$ entirely while still guaranteeing sublinear regret, avoiding any exponential scaling with the game's parameters.

To formalize the strategy, we first define the count matrix $\actioncount_\timestep \in \mathbb{Z}^{\actionsize \times \actionsize}$, where $\actioncount_\timestep(\action^1,\action^2)$ denotes the number of times the joint action $(\action^1,\action^2)$, the learner playing $\action^1$ and the adversary playing $\action^2$, has occurred up to time $\timestep$ (noninclusive). 
We further define the \textbf{discrepancy matrix} $\actiondelta_\timestep$ as the difference between the counts of symmetric action pairs:
$
    \actiondelta_\timestep(\action^1,\action^2) \coloneqq \actioncount_\timestep(\action^2,\action^1) - \actioncount_\timestep(\action^1,\action^2).
$

We analyze the strategy where the learner plays the original Copycat algorithm of \cite{feldman2010playing} on each state $\state \in \State$ independently. Let $\utility_\star$ denote the true (unobserved) payoff of the game. We define the cumulative payoff $\alpha$ as:
$
\alpha = \sum_{\timestep=1}^\Time\sum_{\horizon=1}^\Horizon\utility_\star(\state_{\horizon,\timestep}, \action^1_{\horizon,\timestep},\action^2_{\horizon,\timestep}).
$
Using the fact that the magnitude of the cumulative payoff is bounded by the final discrepancy of the action counts. Specifically: 
$\abs*{\alpha} \leq\sum_{\horizon=1}^\Horizon\sum_{\state\in\State_\horizon}\sum_{\action^1,\action^2 \in \Action}\frac{1}{2}\abs*{\actiondelta_{\Time+1}(\state, \action^1,\action^2)}
$
where $\actiondelta_{\Time+1}(\state, \action^1, \action^2)$ is the final entry in the discrepancy matrix for state $\state$ after $\Time$ episodes.

and,by Cauchy-Schwartz:
\begin{align*}
    (\E\brk*{\abs*{\alpha}})^2 \leq \E\brk[s]*{\alpha^2}
    &\leq \E\brk[s]*{\brk*{\sum_{\horizon=1}^\Horizon\sum_{\state\in\State_\horizon}\sum_{\action^1,\action^2\in\Action}\frac{1}{2}\abs*{\actiondelta_{\Time+1}(\state, \action^1,\action^2)}}^2}\\
    &\leq \E\brk[s]*{\frac{|\State|\actionsize^2}{4}\sum_{\horizon=1}^\Horizon\sum_{\state\in\State_\horizon}\sum_{\action^1,\action^2\in\Action}(\actiondelta_{\Time+1}(\state, \action^1,\action^2))^2}
\end{align*}
To analyze the growth of this term, we define the potential function $\beta^\horizon_\timestep$ representing the sum of squared discrepancies at time $\timestep$ for states at horizon $\horizon$:
$
    \beta^\horizon_\timestep 
    = 
    \sum_{\state\in\State_\horizon}\sum_{\action^1,\action^2\in\Action}(\actiondelta_{\timestep}(\state, \action^1,\action^2))^2.
$

We analyze the single-step difference $\beta^\horizon_{\timestep+1} - \beta^\horizon_\timestep$. Note that in each episode $\timestep$, exactly one state $\state_\horizon$ is visited at each step $\horizon$. For the unvisited states, the discrepancy does not change. For the visited state $s=\state_\horizon$, the change is bounded as follows:
\begin{align*}
    \beta^\horizon_{\timestep+1} - \beta^\horizon_\timestep 
    &= 
    \begin{cases}
        0 & \text{if } \action^1_{\horizon,\timestep}=\action^2_{\horizon,\timestep}; 
        \\
        2-4\actiondelta_\timestep(\state_{\horizon,\timestep}, \action^1_{\horizon,\timestep},\action^2_{\horizon,\timestep}) & \text{otherwise.}
    \end{cases}\\
    &\leq 
    2-4\actiondelta_\timestep(\state_{\horizon,\timestep}, \action^1_{\horizon,\timestep},\action^2_{\horizon,\timestep}). 
\end{align*}

Summing over the horizon $\Horizon$:
\begin{align*}
    \sum_{\horizon=1}^\Horizon\beta^\horizon_{\timestep+1} - \beta^\horizon_\timestep
    &\leq \sum_{\horizon=1}^\Horizon2 -4\actiondelta_\timestep(\state_{\horizon,\timestep}, \action^1_{\horizon,\timestep},\action^2_{\horizon,\timestep})
    = 
    2\Horizon -4\sum_{\horizon=1}^\Horizon\actiondelta_\timestep(\state_{\horizon,\timestep}, \action^1_{\horizon,\timestep},\action^2_{\horizon,\timestep}).
\end{align*}
Crucially, the Copycat strategy computes the max-min strategy for the virtual game defined by $\actiondelta_\timestep(\state,\cdot,\cdot)$. This guarantees that the expected ``virtual payoff'' is non-negative at every step:
$
\E\brk[s]{\actiondelta_\timestep(\state_\horizon, \action^1_\timestep,\action^2_\timestep)} \geq 0.
$
By linearity of expectation:
$
    \E\brk[s]{\sum_{\horizon=1}^\Horizon\actiondelta_\timestep(\state_{\horizon,\timestep}, \action^1_{\horizon,\timestep},\action^2_{\horizon,\timestep})} 
    = 
    \sum_{\horizon=1}^\Horizon\E\brk[s]{\actiondelta_\timestep(\state_{\horizon,\timestep}, \action^1_{\horizon,\timestep},\action^2_{\horizon,\timestep})} \geq 0.
$

Substituting this back into the potential function inequality, we obtain: 
$
    \E\brk[s]{\sum_{\horizon=1}^\Horizon\beta^\horizon_{\timestep+1}-\beta^\horizon_\timestep} \leq 2\Horizon.
$

Since $\beta^\horizon_1 = 0$ (assuming zero initial discrepancy), telescoping the sum over $\Time$ yields $\E\brk[s]{\sum_{\horizon=1}^\Horizon \beta^\horizon_{\Time+1}} \leq 2\Horizon\Time$.
Finally, substituting this bound into our expression for the squared expected payoff:
\begin{align*}
    (\E\brk*{\abs*{\alpha}})^2 
    \leq 
    \E\brk[s]*{\frac{|\State|\actionsize^2}{4}\sum_{\horizon=1}^\Horizon\beta_{\Time+1}^\horizon}
    \leq \frac{\Time|\State|\Horizon\actionsize^2}{2}.
\end{align*}
Taking the square root yields the desired bound. This result confirms that in the SSG setting, the adversary's ability to influence the trajectory yields minimal advantage, as they cannot force a worse regret than what is incurred by solving the local matrix games.

\section{Robustness to Approximate Symmetry ($\epsilon$-Stability)}
\label{app:approximate-symmetry}
A natural question is whether our structural reductions and learning guarantees are robust to small deviations from perfect symmetry. 
If one attempts to push $\epsilon$-inequalities directly through the recursive peeling of the EPR transformation, the exact equalities required for the structural collapse appear to break, as local approximation errors compound iteratively over the horizon.However, we can bypass this algebraic fragility by viewing the problem from a different perspective: reframing the $\epsilon$-approximation as a global model misspecification. 
Instead of defining the relaxation purely through local $\epsilon$-inequalities over the value function, we define it as the distance to the nearest exact game. 
Let $\Utility$ be the space of exact MSG (or HSG) payoffs. We assume the true payoff $\utility^\star$ is $\epsilon$-close to some exact game $\utility \in \Utility$, such that:
\begin{equation}
\label{eq:eppsilon-assum}
||\utility^\star - \utility||_1 \le \epsilon    
\end{equation}
Note that these two definitions are essentially interchangeable, merely yielding different scalings of $\epsilon$.
We now prove the following:
\begin{proposition}
\label{prop:epsilon_stability}
If the true game $\utility^\star$ is $\epsilon$-close to an exact symmetric game $\utility \in \Utility$, deploying our algorithm guarantees that the expected cumulative regret degrades gracefully, bounded by $O(\Regret(\utility) + \epsilon \Time)$, where $\Regret(\utility)$ is the regret of the game under one of our propositions.
Furthermore, the learner cannot escape the $\Omega(\epsilon\Time)$ factor.
\end{proposition}
\begin{proof}
When the learner deploys our algorithm, they can guarantee that for any $\utility\in\Utility$, $\text{Regret}(\utility)$ is bounded by the relevant theorem (\ref{prop:ssg}, \ref{prop:msg}, \ref{prop:hsg}).

Lets assume $\hat\utility = \argmin_{\utility\in\Utility}||\utility^\star - \utility||_{1}$ and define $\Delta\utility = \utility^\star - \hat\utility$, then:
\begin{align*}
    \E\brk[s]*{\abs*{\sum_{\timestep=1}^\Time\sum_{\horizon=1}^\Horizon\utility^\star(\state_{\horizon, \timestep},\action^1_{\horizon, \timestep},\action^2_{\horizon, \timestep})}} &= \E\brk[s]*{\abs*{\sum_{\timestep=1}^\Time\sum_{\horizon=1}^\Horizon\hat\utility(\state_{\horizon, \timestep},\action^1_{\horizon, \timestep},\action^2_{\horizon, \timestep}) + \Delta\utility(\state_\horizon,\action^1,\action^2)}}\\
        &\leq \E\brk[s]*{\abs*{\sum_{\timestep=1}^\Time\sum_{\horizon=1}^\Horizon\hat\utility(\state_{\horizon, \timestep},\action^1_{\horizon, \timestep},\action^2_{\horizon, \timestep})}} + \E\brk[s]*{\abs*{{\sum_{\timestep=1}^\Time\sum_{\horizon=1}^\Horizon\Delta\utility(\state_{\horizon, \timestep},\action^1_{\horizon, \timestep},\action^2_{\horizon, \timestep})}}}\tag{triangle inequality}\\
        &\leq \E\brk[s]*{\abs*{\sum_{\timestep=1}^\Time\sum_{\horizon=1}^\Horizon\hat\utility(\state_{\horizon, \timestep},\action^1_{\horizon, \timestep},\action^2_{\horizon, \timestep})}} + \Time\epsilon \tag{\cref{eq:eppsilon-assum}}\\
        &\leq \Regret(\hat\utility) + \Time\epsilon\\
\end{align*}

Crucially, because our setting assumes the learner cannot observe payoffs, they remain completely blind to the specific direction and structure of this $\epsilon$-perturbation. From a worst-case guarantee perspective, this is mathematically equivalent to an adversary selecting the perturbed payoff function in hindsight. 
To illustrate why this $\Omega(\epsilon\Time)$ bound is tight and cannot be improved by the learner, let us consider a fully symmetric SSG with the payoff function $\utility$ where deviating in $\state_1$ incurs a severe penalty $\utility(\state_1,\action^1\neq 1,\action^2=1)=-1$. Now, we choose the following $\epsilon$-perturbation:
\begin{align*}
    \utility^\star = \begin{cases}
        -\epsilon &  \text{if } (\state,\action^1,\action^2) = (\state_1,1,1)\\
        \utility & \text{otherwise }
    \end{cases}
\end{align*}
Consider an adversary that always plays action $1$ at state $\state_1$. The optimal policy for the learner is to play action $1$. By doing so, they inevitably absorb $-\epsilon$ penalty, which is still strictly better than playing any other action and triggering the catastrophic $-1$ penalty.
Consequently,
\[
\Regret(\utility^\star)\ge\epsilon\Time.
\]
Therefore, the additive $\epsilon\Time$ term is
unavoidable.

\end{proof}

\section{Examples}
\subsection{SSG but not MSG}
\label{sec:SSG-not-MSG}
    
\begin{figure}[H]
\centering
\begin{tikzpicture}[
    scale=0.75, 
    transform shape,
    game matrix/.style={
        draw,
        thick,
        fill=white,
        rectangle,
        rounded corners=2pt,
        inner sep=5pt,
        align=center
    },
    transition arrow/.style={
        ->,
        >=stealth, 
        thick,
        color=black!70
    },
    label style/.style={
        font=\sffamily\small,
        color=black!80
    },
    title style/.style={
        font=\sffamily\bfseries\Large,
        color=black
    }
]

    \newcommand{\payoffmatrix}[4]{
        \begin{tabular}{r|c|c|}
            \multicolumn{1}{r}{} & \multicolumn{1}{c}{$1$} & \multicolumn{1}{c}{$2$} \\ \cline{2-3}
            $1$ & #1 & #2 \\ \cline{2-3}
            $2$ & #3 & #4 \\ \cline{2-3}
        \end{tabular}
    }

    
    \node[game matrix](state1) at (0,0) {
        \textbf{State $s_1$}\\
        \payoffmatrix{0}{1}{-1}{0}
    };

    \node[game matrix](state2) at (6, 1.2) {
        \textbf{State $s_2^A$}\\
        \payoffmatrix{0}{1}{-1}{0}
    };

    \node[game matrix](state3) at (6, -1.2) {
        \textbf{State $s_2^B$}\\
        \payoffmatrix{0}{-1}{1}{0}
    };

    
    \draw[transition arrow] (state1.east) -- (state2.west) 
        node[midway, above, sloped, label style] {$\p(\state^A_2\mid\state_1,a_1, a_2)$};

    \draw[transition arrow] (state1.east) -- (state3.west) 
        node[midway, below, sloped, label style] {$\p(\state^B_2\mid\state_1,a_1, a_2)$};

\end{tikzpicture}
\caption{Example: SSG but not MSG game.}
\label{fig:ssg-not-msg}
\end{figure}
In this example, if $\transition(\state_2^A\mid\state_1,1,2)=\transition(\state_2^B\mid\state_1,2,1) = 1$, 
consider the deterministic strategies:
\begin{align*}
    &\strategy^1(\state) = \begin{cases}
    1, &\state = \state_1\\
    1, &\state =  \state_2^A\\
    1, &\state =  \state_2^B        
    \end{cases}
    & \strategy^2(\state) = \begin{cases}
    2, &\state =  \state_1\\
    1, &\state =  \state_2^A\\
    2, & \state =  \state_2^B        
    \end{cases}
\end{align*}
where $\strategy^i:\State\mapsto\Action$ denotes the action selected by player $i$ at each state.
Here, $\Value^{\strategy^1,\strategy^2}(\state_1) = 1 + 0 = 1$, whereas $\Value^{\strategy^2,\strategy^1}(\state_1) = -1 + 1 = 0$. The game value is asymmetric, violating MSG, despite every local payoff matrix being skew-symmetric. 

This counter-example clearly demonstrates that the MSG assumption induces severe structural constraints. When players swap policies (switching $\strategy^1$ and $\strategy^2$), the resulting trajectory of state visitations may change completely. Yet, the MSG condition demands that the total value of the game remains exactly symmetric (negated) despite this potentially drastically different path through the state space. This implies that the payoffs and transition dynamics must be tightly coupled to ensure such invariance holds across all possible Markov policies.

\subsection{MSG but not HSG}
\label{sec:msg-not-hsg}

\begin{figure}[H]
\centering
\begin{tikzpicture}[
    scale=0.75, 
    transform shape,
    game matrix/.style={
        draw,
        thick,
        fill=white,
        rectangle,
        rounded corners=2pt,
        inner sep=5pt,
        align=center
    },
    transition arrow/.style={
        ->,
        >=stealth, 
        thick,
        color=black!70
    },
    label style/.style={
        font=\sffamily\small,
        color=black!80
    },
    title style/.style={
        font=\sffamily\bfseries\Large,
        color=black
    }
]
    \newcommand{\payoffmatrix}[4]{
        \begin{tabular}{r|c|c|}
            \multicolumn{1}{r}{} & \multicolumn{1}{c}{$1$} & \multicolumn{1}{c}{$2$} \\ \cline{2-3}
            $1$ & #1 & #2 \\ \cline{2-3}
            $2$ & #3 & #4 \\ \cline{2-3}
        \end{tabular}
    }

    
    \node[game matrix](state1) at (0,0) {
        \textbf{State $s_1$}\\
        \payoffmatrix{0}{1}{-1}{0}
    };

    \node[game matrix](state2) at (6, 0) {
        \textbf{State $s_2$}\\
        \payoffmatrix{0}{1}{-1}{0}
    };

    
    \draw[transition arrow] (state1.east) -- (state2.west) 
        node[midway, above, sloped, label style] {Action $(a_1, a_2)$};

\end{tikzpicture}
\caption{Example: MSG but not HSG game.}
\label{fig:msg-not-hsg}
\end{figure}

Consider the game above with deterministic transitions $\transition(\state_2\mid\state_1,\cdot,\cdot) = 1$. 
Let us examine the value under specific history-dependent strategies:
\begin{align*}
    &\strategy^1(\history) = \begin{cases}
    1,  &\history =  (\state_1)\\
    1,  &\history =  (\state_1,1,2,\state_2)\\
    1,  &\history =  (\state_1,2,1,\state_2)
    \end{cases};
    &\strategy^2(\history) = \begin{cases}
    2,  &\history =  (\state_1)\\
    1,  &\history =  (\state_1,1,2,\state_2)\\
    2,  &\history =  (\state_1,2,1,\state_2)
    \end{cases},
\end{align*}
where $\strategy^i:\History\mapsto\Action$ denotes the action selected by player $i$ for each history.

Under these strategies, the trajectory is $(\state_1, 1, 2) \to \state_2$. The value is:
   $ \Value^{\strategy^1,\strategy^2}(\state_1) = \utility(\state_1, 1, 2) + \utility(\state_2, 1, 1) = 1 + 0 = 1.$
However, if we swap roles, the trajectory becomes $(\state_1, 2, 1) \to \state_2$, and the opponent responds differently to this history. The value becomes:
$    \Value^{\strategy^2,\strategy^1}(\state_1) = \utility(\state_1, 2, 1) + \utility(\state_2, 2, 1) = -1 - 1 = -2.$
Since $1 \neq -(-2)$, the game violates HSG. 
This example demonstrates that HSG requires symmetry to hold even when strategies condition their future behavior on the specific history of past actions. 
This restriction forces the continuation value to be independent of future policies after a disagreement has occurred.


\newpage

\end{document}